\documentclass[12pt, draftclsnofoot, onecolumn]{IEEEtran}
\IEEEoverridecommandlockouts

\usepackage{cite}
\usepackage{physics}
\usepackage{amsmath,amssymb,amsfonts,amsthm}
\usepackage{lipsum}
\usepackage{mathtools}
\usepackage{cuted}
\usepackage{algorithm}
\usepackage{algorithmicx}
\usepackage{algpseudocode}
\usepackage{graphicx}
\usepackage{caption}
\usepackage{textcomp}
\usepackage{cite}
\usepackage{scalerel}
\usepackage{gensymb}
\usepackage{xcolor}

\newtheorem{theorem}{Theorem}
\newtheorem{prop}{Proposition}

\newtheorem{lemma}{Lemma}
\def\BibTeX{{\rm B\kern-.05em{\sc i\kern-.025em b}\kern-.08em
    T\kern-.1667em\lower.7ex\hbox{E}\kern-.125emX}}
    
\newcommand{\algrule}[1][.2pt]{\par\vskip.3\baselineskip\hrule height #1\par\vskip.3\baselineskip}

\begin{document}

{\Large \textbf{Notice:} This work has been submitted to the IEEE for possible publication. Copyright may be transferred without notice, after which this version may no longer be accessible.}
\clearpage

\title{Downlink Pilot Precoding and Compressed Channel Feedback for FDD-Based Cell-Free Systems\\}
\author{\textsuperscript{*}Seungnyun Kim\thanks{This work was supported by Electronics and Telecommunications Research Institute (ETRI) grant funded by the Korean government [2018-0-01410, Development of Radio Transmission
Technologies for High Capacity and Low Cost in Ultra Dense Networks]. 

Parts of this paper was appeared at the ICC, 2019~\cite{kim2019feedback}.}, \textsuperscript{\textdagger}Jun Won Choi, and \textsuperscript{*}Byonghyo Shim\\
\textsuperscript{*}Seoul National University, \textsuperscript{\textdagger}Hanyang University, Seoul, Korea\\
Email: \textsuperscript{*}snkim@islab.snu.ac.kr, \textsuperscript{\textdagger}junwchoi@hanyang.ac.kr, \textsuperscript{*}bshim@islab.snu.ac.kr}
\maketitle

\begin{abstract}
Cell-free system where a group of base stations (BSs) cooperatively serves users has received much attention as a promising technology for the future wireless systems. In order to maximize the cooperation gain in the cell-free systems, acquisition of downlink channel state information (CSI) at the BSs is crucial. While this task is relatively easy for the time division duplexing (TDD) systems due to the channel reciprocity, it is not easy for the frequency division duplexing (FDD) systems due to the CSI feedback overhead. This issue is even more pronounced in the cell-free systems since the user needs to feed back the CSIs of multiple BSs. In this paper, we propose a novel feedback reduction technique for the FDD-based cell-free systems. Key feature of the proposed technique is to choose a few dominating paths and then feed back the path gain information (PGI) of the chosen paths. By exploiting the property that the angles of departure (AoDs) are quite similar in the uplink and downlink channels (this property is referred to as \textit{angle reciprocity}), the BSs obtain the AoDs directly from the uplink pilot signal. From the extensive simulations, we observe that the proposed technique can achieve more than 80\% of feedback overhead reduction over the conventional CSI feedback scheme.
\end{abstract}

\newpage

\section{Introduction}
In recent years, ultra dense network (UDN) has received a great deal of attention as a means to achieve a thousand-fold throughput improvement in 5G wireless communications~\cite{series2015imt}. Network densification can improve the capacity of cellular systems by overlaying the existing macro cells with a large number of small (femto, pico) cells. However, throughput improvement of dense networks might not be dramatic as expected due to the poor cell-edge performance. This is because the portion of users in the cell-boundary (cell-edge users) increases sharply yet cell-edge users suffer from significant inter-cell interference due to the reduced cell size. To address this problem, an approach to entirely remove the notion of cell from the cellular systems, called \textit{cell-free} systems, has been introduced recently~\cite{ngo2017cell}. When compared to the conventional cellular systems in which a single base station (BS) serves all the users in a cell, a group of BSs cooperatively serves users in the cell-free systems (see Fig. 1). In the cell-free systems, BSs are connected to the digital unit (DU) via advanced backhaul links to share the channel state information (CSI) and the transmit data. Since the cell association is not strictly limited by the regional cell, notions like \textit{cell} and \textit{cell boundary} are unnecessary in the cell-free systems. Also, since the DU intelligently recognizes the user's communication environments and organizes the associated BSs for each user, cell-free systems can control inter-cell interference efficiently, thereby achieving significant improvement in the spectral efficiency and coverage.

\begin{figure}[t]
\centering
\includegraphics[scale=0.6]{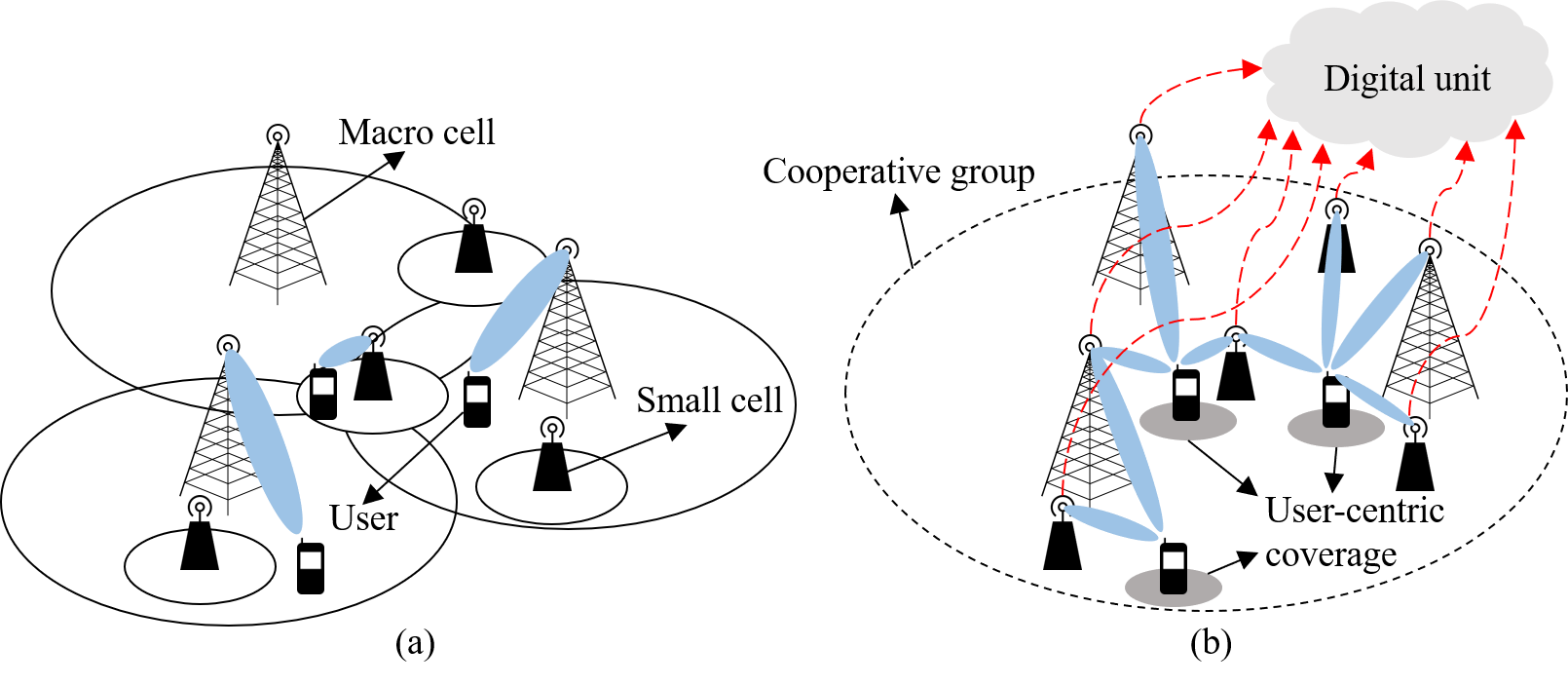}
\caption{Comparison between (a) the conventional cellular systems and (b) the cell-free systems.}
\end{figure} 

In order to maximize the gain obtained by the BS cooperation, acquisition of accurate downlink CSI at the BS is crucial. While this task is relatively easy for the time division duplexing (TDD) systems due to the channel reciprocity, it is not easy for the frequency division duplexing (FDD) systems due to the CSI feedback overhead. For this reason, most efforts on the cell-free systems to date are based on the TDD systems~\cite{ngo2017cell, nayebi2017precoding, ngo2018total}. In practice, however, TDD-based cell-free systems have some potential problems. For example, due to the switching between the uplink and downlink transmission in the TDD systems, users may not be able to obtain the instantaneous CSI when the transmission direction is directed to the uplink~\cite{jose2011pilot}. Further, the channel reciprocity in TDD systems might not be accurate due to the calibration error in the RF chains~\cite{lee2015antenna}. These observations, together with the fact that the FDD systems have many benefits over the TDD systems (e.g., continuous channel estimation and small latency), motivate us to study FDD-based cell-free systems. One well-known drawback of the FDD systems is that the amount of CSI feedback needs to be proportional to the number of transmit antennas to achieve the rate comparable to the system with the perfect CSI~\cite{jindal2006mimo}. This issue is even more pronounced in the cell-free systems since the user needs to estimate and feed back the downlink CSIs of multiple BSs. Therefore, it is of a great importance to come up with an effective means to relax the feedback overhead in the FDD-based cell-free systems.

The primary purpose of this paper is to propose an approach to reduce the CSI feedback overhead in the FDD-based cell-free systems. Key feature of the proposed technique is that the spatial domain channel can be represented by a small number of multi-path components (angle of departure (AoD) and path gain)~\cite{ertel1998overview}. By exploiting the property referred to as \textit{angle reciprocity}~\cite{xie2016overview} that the AoDs are quite similar in the uplink and downlink channels, we only feed back the path gain information (PGI) to the BSs. As a result, the number of bits required for the channel vector quantization scales linearly with the number of dominating paths, not the number of transmit antennas. Moreover, by choosing a few dominating paths maximizing the sum rate, we can further reduce the feedback overhead considerably. In order to support the dominating PGI acquisition and feedback at the user, we use spatially precoded downlink pilot signal.

Through the performance analysis, we show that the proposed dominating PGI feedback scheme exhibits a smaller quantization distortion than that generated by the conventional CSI feedback scheme. In fact, the number of feedback bits required to maintain a constant gap to the system with perfect PGI scales linearly with the number of dominating paths which is much smaller than the number of transmit antennas. From the simulations on realistic scenarios, we show that the proposed dominating PGI feedback scheme achieves more than 80\% of feedback overhead reduction over the conventional scheme relying on the CSI feedback. We also show that the performance gain of the proposed dominating PGI feedback scheme increases with the number of cooperating BSs. Note that no such benefit can be obtained for the conventional CSI feedback scheme from the BS cooperation. This implies that the proposed dominating PGI feedback scheme is a promising solution to reduce the feedback overhead in FDD-based cell-free systems.

The rest of this paper is organized as follows. In Section II, we briefly introduce the system and channel models for FDD-based cell-free systems. In Section III, we present the dominating path selection technique. In Section IV, we present the downlink pilot precoding scheme for the dominating PGI acquisition. In Section V, we present the performance analysis of the proposed dominating PGI feedback scheme. In Section VI, we present the simulation results and conclude the paper in Section VII.

\textit{Notations}: Lower and upper case symbols are used to denote vectors and matrices, respectively. The superscripts $(\cdot)^{\textrm{T}}$, $(\cdot)^{\textrm{H}}$, and $(\cdot)^{+}$ denote transpose, Hermitian transpose, and pseudo-inverse, respectively. $\otimes$ denotes the Kronecker product. $\left\lVert\mathbf{x}\right\rVert$ and $\left\lVert\mathbf{X}\right\rVert_{\text{F}}$ are used as the Euclidean norm of a vector $\mathbf{x}$ and the Frobenius norm of a matrix $\mathbf{X}$, respectively. $\text{tr}\left(\mathbf{X}\right)$ and $\text{vec}\left(\mathbf{X}\right)$ denote the trace and vectorization of $\mathbf{X}$, respectively. Also, $\text{diag}\left(\mathbf{X}_{1},\mathbf{X}_{2}\right)$ denotes a block diagonal matrix whose diagonal elements are $\mathbf{X}_{1}$ and $\mathbf{X}_{2}$. In addition, $\mathbf{x}_{\Lambda}$ is a subvector of $\mathbf{x}$ whose $i$-th entry is $\mathbf{x}(\Lambda(i))$ and $\mathbf{X}_{\Lambda}$ is a submatrix of $\mathbf{X}$ whose $i$-th column is the $\Lambda(i)$-th column of $\mathbf{X}$ for $i=1,\cdots,\left\lvert\Lambda\right\rvert$ ($\Lambda$ is the set of partial indices and $\left\lvert\Lambda\right\rvert$ is the cardinality of $\Lambda$).
\\

\section{cell-free System Model}
In this section, we introduce the FDD-based cell-free systems and the multi-path channel model. We also discuss the angle reciprocity between the uplink and downlink channels and the conventional quantized channel feedback scheme. 

\begin{figure}[t]
\centering
\includegraphics[scale=0.7]{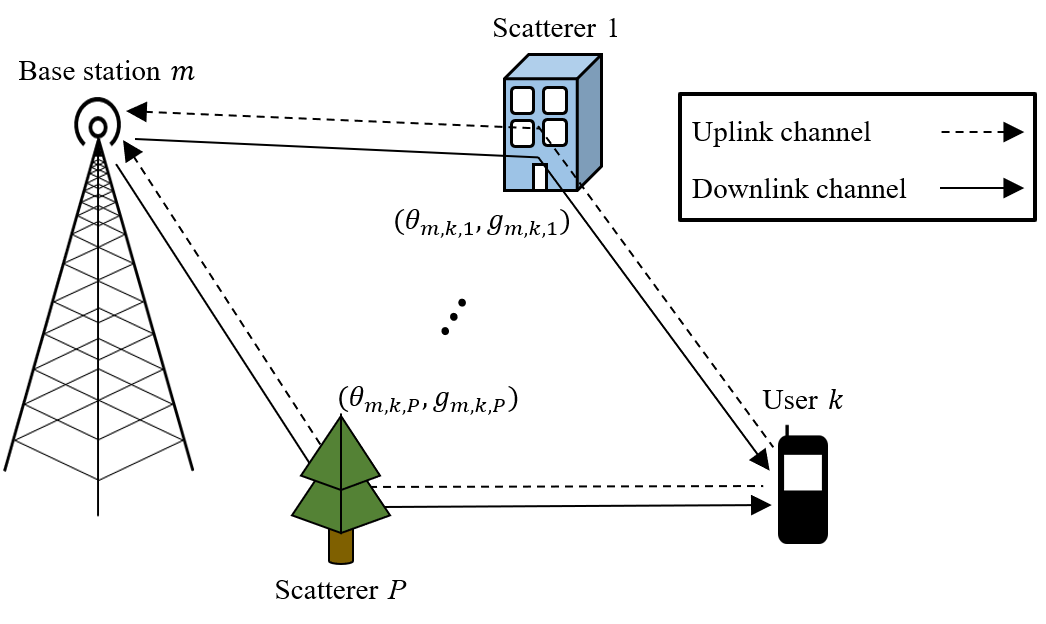}
\caption{Narrowband ray-based channel model and angle reciprocity between the uplink and downlink channels.}
\end{figure}

\subsection{cell-free System Model}
We consider the FDD-based cell-free systems with $M$ BSs and $K$ users. Each BS is equipped with a uniform linear array of $N$ antennas and each user is equipped with a single antenna. Let $\mathcal{B}=\lbrace 1,\cdots, M\rbrace$ and $\mathcal{U}=\lbrace 1,\cdots,K\rbrace$ be the sets of BSs and users, respectively. In our work, we consider the narrowband ray-based channel model consisting of $P$ paths (see Fig. 2)~\cite{tse2005fundamentals}. The downlink channel vector $\mathbf{h}_{m,k}\in\mathbb{C}^{N}$ from the BS $m$ to the user $k$ is expressed as
\begin{align}\label{2.1.1}
\mathbf{h}_{m,k}=\sum_{i=1}^{P}g_{m,k,i}\mathbf{a}\left(\theta_{m,k,i}\right),
\end{align}
where $\theta_{m,k,i}$ is the AoD and $g_{m,k,i}$ is the complex path gain of the $i$-th path, respectively. We assume that for every $m$, $k$, and $i$, $g_{m,k,i}\sim\mathcal{CN}(0,1)$ are independent and identically distributed (i.i.d.) random variables. In addition, $\mathbf{a}\left(\theta_{m,k,i}\right)\in \mathbb{C}^{N}$ is the array steering vector given by
\begin{equation}\label{2.1.2}
\mathbf{a}\left(\theta_{m,k,i}\right)=\left[1,e^{-j2\pi \frac{d}{\lambda}\sin{\theta_{m,k,i}}},\,\cdots,\,e^{-j(N-1)2\pi \frac{d}{\lambda}\sin{\theta_{m,k,i}}}\right]^{\textrm{T}},
\end{equation}
where $d$ is the antenna spacing and $\lambda$ is the signal wavelength. The matrix-vector form of $\mathbf{h}_{m,k}$ is
\begin{align}\label{2.1.3}
\mathbf{h}_{m,k}=\mathbf{A}_{m,k}\mathbf{g}_{m,k},
\end{align}
where $\mathbf{A}_{m,k}=\left[\mathbf{a}\left(\theta_{m,k,1}\right),\cdots,\mathbf{a}\left(\theta_{m,k,P}\right)\right]\in \mathbb{C}^{N\times P}$ is the array steering matrix and $\mathbf{g}_{m,k}=\left[g_{m,k,1},\cdots,g_{m,k,P}\right]^{\textrm{T}}\in \mathbb{C}^{P}$ is the PGI vector. It is worth mentioning that the AoDs vary much slower than the path gains. In fact, since scatterers affecting the signal transmission do not change their positions significantly, the AoDs are readily considered as constant during the channel coherence time. Also, it has been shown that the number of propagation paths $P$ is quite smaller than the number of transmit antennas $N$~\cite{shen2016joint}. We note that $P$ is completely determined by the scattering geometry around the BS. Since the BSs are usually located at high places such as a rooftop of a building, only a few scatterers affect the signal transmission. For example, $P$ is $2\sim 8$ for $6\sim 60\,\text{GHz}$ band due to the limited scattering of the millimeter-wave signal~\cite{rappaport2013millimeter}. Also, for the sub-$6\,\text{GHz}$ band, $P$ is set to $10\sim 20$ (3GPP spatial channel model~\cite{3GPPTR}) while $N$ is $32\sim 256$ in the massive multiple-input multiple-output (MIMO) regime. In this setting, the received signal $y_{k}\in\mathbb{C}$ of the user $k$ is given by
\begin{align}\label{2.1.4}
y_{k}=\sum_{m=1}^{M}\mathbf{h}_{m,k}^{\textrm{H}}\mathbf{w}_{m,k}s_{k}+\sum_{j\neq k}^{K}\sum_{m=1}^{M}\mathbf{h}_{m,k}^{\textrm{H}}\mathbf{w}_{m,j}s_{j}+n_{k},
\end{align}
where $\mathbf{w}_{m,k}\in\mathbb{C}^{N}$ is the precoding vector from the BS $m$ to the user $k$, $s_{k}\in\mathbb{C}$ is the data symbol for the user $k$, and $n_{k}\sim\mathcal{CN}(0,\sigma_{n}^{2})$ is the additive Gaussian noise. The corresponding achievable rate $R_{k}$ of the user $k$ is 
\begin{align}\label{2.1.5}
R_{k}=\mathbb{E}\left[\log_{2}\left(1+\frac{\left\lvert\sum_{m=1}^{M}\mathbf{h}_{m,k}^{\textrm{H}}\mathbf{w}_{m,k}\right\rvert^{2}}{\sum_{j\neq k}^{K}\left\lvert\sum_{m=1}^{M}\mathbf{h}_{m,k}^{\textrm{H}}\mathbf{w}_{m,j}\right\rvert^{2}+\sigma_{n}^{2}}\right)\right].
\end{align}
Approximately, we have\footnote{This approximation becomes more accurate as the number of transmit antennas $N$ increases~\cite[Lemma 1]{zhang2015power}.}
\begin{align}\label{2.1.6}
R_{k}\approx\log_{2}\left(1+\frac{\mathbb{E}\left[\left\lvert\sum_{m=1}^{M}\mathbf{h}_{m,k}^{\textrm{H}}\mathbf{w}_{m,k}\right\rvert^{2}\right]}{\sum_{j\neq k}^{K}\mathbb{E}\left[\left\lvert\sum_{m=1}^{M}\mathbf{h}_{m,k}^{\textrm{H}}\mathbf{w}_{m,j}\right\rvert^{2}\right]+\sigma_{n}^{2}}\right).
\end{align}

\subsection{Angle Reciprocity between Uplink and Downlink Channels}
As mentioned, the AoDs in the uplink and downlink channels are fairly similar in the FDD systems when their carrier frequencies do not differ too much (typically less than a few GHz). The reason is because only the signal components that physically reverse the uplink propagation path can reach the user during the downlink transmission~\cite{xie2016overview} (see Fig. 2). Since the changes of relative permittivity and conductivity of the scatterers are negligible in the scale of several GHz, reflection and deflection properties determining the propagation paths in the uplink and downlink transmissions are fairly similar~\cite{series2012propagation}, which in turn implies that the propagation paths of the uplink and downlink channels are more or less similar. This so-called \textit{angle reciprocity} is very useful since the BS can acquire the AoDs from the uplink pilot signal. In estimating the AoDs, various algorithms such as multiple signal classification (MUSIC)~\cite{schmidt1986multiple} or estimation of signal parameters via rotational invariance techniques (ESPRIT)~\cite{roy1989esprit} can be employed.

\subsection{Conventional Quantized Channel Feedback}
In the conventional quantized channel feedback, a user estimates the downlink channel vector from the downlink pilot signal. Then, the user quantizes the channel direction $\bar{\mathbf{h}}_{m,k}=\frac{\mathbf{h}_{m,k}}{\lVert \mathbf{h}_{m,k}\rVert}$ and then feeds it back to the BS. Specifically, a codeword $\mathbf{c}_{\hat{i}_{m,k}}$ is chosen from a pre-defined $B$-bit codebook $\mathcal{C}=\lbrace\mathbf{c}_{1},\cdots,\mathbf{c}_{2^{B}}\rbrace$ as
\begin{align}\label{2.3.1}
\mathbf{c}_{\hat{i}_{m,k}}=\text{arg }\underset{\mathbf{c}\in\mathcal{C}}{\text{max}}\,\lvert\bar{\mathbf{h}}_{m,k}^{\textrm{H}}\mathbf{c}\rvert^{2}.
\end{align} 
Then, the selected index $\hat{i}_{m,k}$ is fed back to the BS. It has been shown that the number of feedback bits $B$ needs to be scaled linearly with the channel dimension $N$ and SNR (in decibels) to properly control the quantization distortion as~\cite{jindal2006mimo} 
\begin{align}\label{2.3.2}
B\approx\frac{(N-1)}{3}\times\text{SNR}.
\end{align}
In the FDD-based cell-free systems, since multiple BSs cooperatively serve users, a user should send the downlink CSIs to multiple BSs. Thus, the feedback overhead should also increase with the number of associated BSs $M$. For example, if $M=6$, $N=16$, and $\text{SNR}=10\,\text{dB}$, then a user has to send $B=300$ bits just for the CSI feedback.
\\

\section{Dominating Path Gain Information Feedback in cell-free Systems}
The key idea of the proposed dominating PGI feedback scheme is to select a small number of paths based on the AoD information and then feed back the measured path gains of the chosen paths. As mentioned, the AoDs are acquired from the uplink pilot signal by using the angle reciprocity. Since the number of propagation paths is smaller than the number of transmit antennas, we can achieve a considerable reduction in the quantized channel dimension using the dominating PGI feedback. We can further reduce the feedback overhead from multiple BSs by choosing a few dominating paths among all possible multi-paths. 

\begin{figure}[t]
\centering
\includegraphics[scale=0.31]{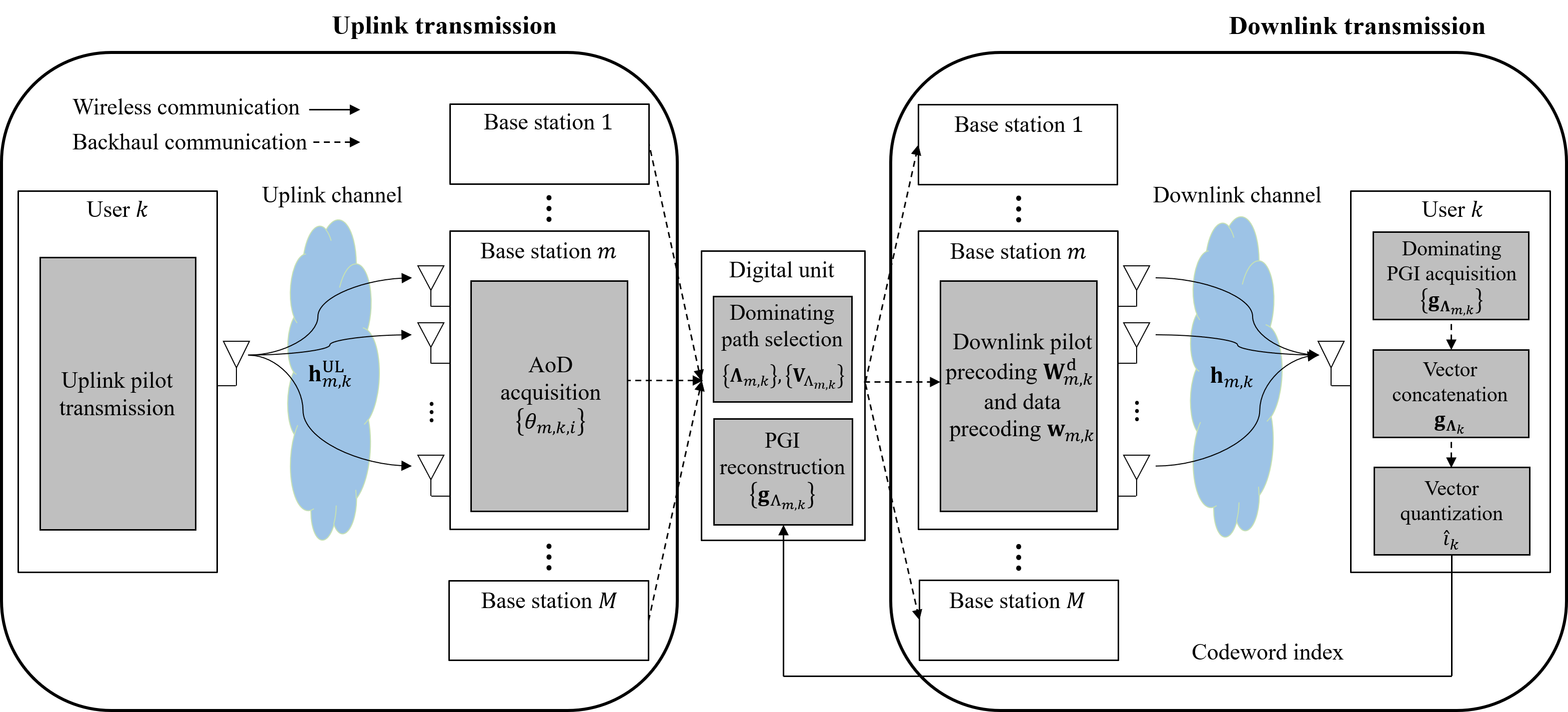}
\caption{Overall transceiver structure of the proposed dominating PGI feedback scheme.}
\end{figure} 

In a nutshell, overall operations of the proposed dominating PGI feedback scheme are as follows: 1) user transmits the uplink pilot signal and then BSs acquire AoDs from the uplink pilot signal, 2) DU performs the dominating path selection based on the acquired AoDs, 3) BSs transmit the precoded downlink pilot signal, 4) each user acquires the dominating PGI from the precoded downlink pilot signal and then feeds it back to the BSs, and 5) BSs perform the downlink data transmission based on the dominating PGI feedback (see Fig. 3).

\subsection{Uplink AoD Acquisition}
Since the AoDs are quite similar in the uplink and downlink channels, the BS can acquire the AoD information from the uplink pilot signal. Well-known angle estimation algorithm includes MUSIC~\cite{schmidt1986multiple} and ESPRIT~\cite{roy1989esprit}. In the MUSIC algorithm, for example, the BS estimates the uplink channel vector $\mathbf{h}_{m,k}^{\text{UL}}$ and then computes the channel covariance matrix $\mathbf{R}_{m,k}^{\text{UL}}=\mathbb{E}\left[\mathbf{h}_{m,k}^{\text{UL}}\mathbf{h}_{m,k}^{\text{UL},\textrm{H}}\right]$. Key idea of the MUSIC algorithm is to decompose the eigenspace of $\mathbf{R}_{m,k}^{\text{UL}}$ into two orthogonal subspaces: signal subspace and noise subspace. To be specific, the eigenvectors of $\mathbf{R}_{m,k}^{\text{UL}}$ that correspond to the $P$ largest eigenvalues form the signal subspace matrix $\mathbf{E}_{s}$ and the rest form the noise subspace matrix $\mathbf{E}_{n}$. Since $\mathbf{E}_{n}$ is orthogonal to the signal subspace, the AoD $\theta$ should satisfy $\mathbf{E}_{n}^{\textrm{H}}\mathbf{a}\left(\theta\right)=\mathbf{0}_{P}$. Thus, the AoDs are obtained from the peak of spectrum function $f_{\text{MUSIC}}(\theta)$ given by
\begin{align}\label{3.1.1}
f_{\text{MUSIC}}(\theta)=\frac{1}{\mathbf{a}^{\textrm{H}}\left(\theta\right)\mathbf{E}_{n}\mathbf{E}_{n}^{\textrm{H}}\mathbf{a}\left(\theta\right)}.
\end{align}

\subsection{Dominating Path Selection Problem Formulation}
Main advantage of the dominating PGI feedback over the conventional CSI feedback is the reduction of the channel vector dimension to be quantized. However, since the user should feed back the PGI to multiple BSs, feedback overhead is still considerable. In the proposed scheme, by choosing a few dominating paths among all possible multi-paths between each user and the associated BSs, we can control the feedback overhead at the expense of marginal degradation in the sum rate.

\begin{figure}[t]
\centering
\includegraphics[scale=0.53]{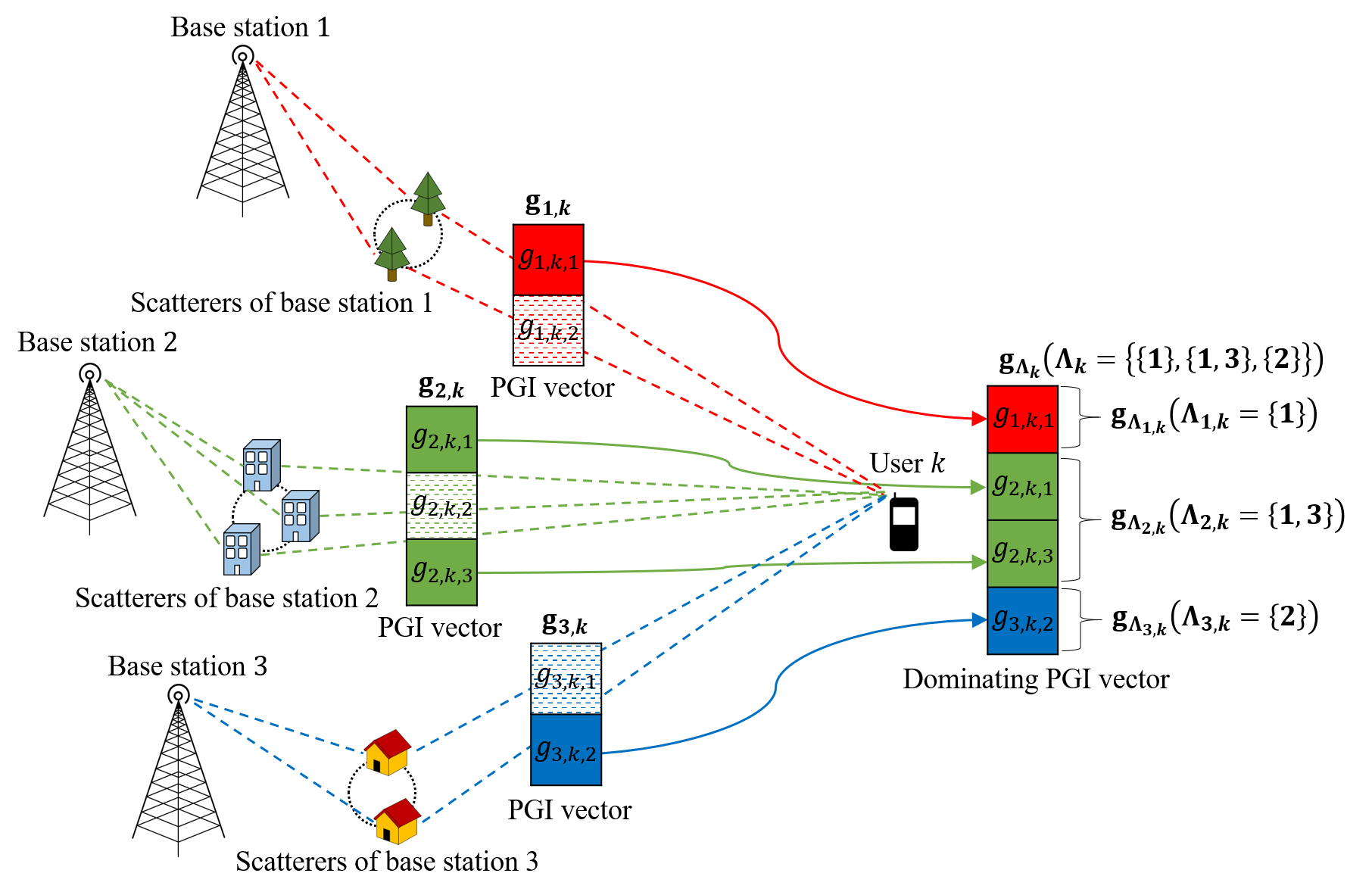}
\caption{Illustration of the dominating path selection}
\end{figure} 

In order to choose the paths that contribute to the sum rate most, we first need to express the sum rate as a function of the dominating paths. Let $\Lambda_{m,k}\subseteq\lbrace 1,\cdots,P\rbrace$ be the index set of the dominating paths from the BS $m$ to the user $k$ and $\mathbf{g}_{\Lambda_{m,k}}=\left[g_{m,k,i},\,i\in\Lambda_{m,k}\right]^{\textrm{T}}\in\mathbb{C}^{\lvert\Lambda_{m,k}\rvert}$ be the dominating PGI vector. For example, if the first and the third paths are chosen as the dominating paths, then $\Lambda_{m,k}=\lbrace 1,3\rbrace$ and $\mathbf{g}_{\Lambda_{m,k}}=\left[g_{m,k,1},\,g_{m,k,3}\right]^{\textrm{T}}$. Also, let $\Lambda_{k}=\lbrace \Lambda_{1,k},\cdots,\Lambda_{M,k}\rbrace$ be the combined index set for the user $k$ and $\mathbf{g}_{\Lambda_{k}}=\left[\mathbf{g}_{\Lambda_{1,k}}^{\textrm{T}},\,\cdots,\mathbf{g}_{\Lambda_{M,k}}^{\textrm{T}}\right]^{\textrm{T}}\in\mathbb{C}^{L}$ be the corresponding dominating PGI vector. Note that $L$ is the total number of dominating paths for each user. For example, if $M=3$, $L=4$, and $\Lambda_{1,k}=\lbrace 1\rbrace$, $\Lambda_{2,k}=\lbrace 1,3\rbrace$, and $\Lambda_{3,k}=\lbrace 2\rbrace$, then $\Lambda_{k}=\lbrace\lbrace 1\rbrace, \lbrace 1,3\rbrace, \lbrace 2\rbrace\rbrace$ and $\mathbf{g}_{\Lambda_{k}}=\left[g_{1,k,1},\,g_{2,k,1},\,g_{2,k,3},\,g_{3,k,1}\right]^{\textrm{T}}$ (see Fig. 4). Then, the user $k$ estimates and feeds back $\mathbf{g}_{\Lambda_{k}}$ to the DU. The downlink precoding vector $\mathbf{w}_{m,k}\in\mathbb{C}^{N}$ from the BS $m$ to the user $k$, constructed from the dominating PGI feedback, is
\begin{align}\label{3.2.1} 
\mathbf{w}_{m,k}=\mathbf{V}_{\Lambda_{m,k}}\hat{\mathbf{g}}_{\Lambda_{m,k}},
\end{align}
where $\mathbf{V}_{\Lambda_{m,k}}\in\mathbb{C}^{N\times\lvert\Lambda_{m,k}\rvert}$ is the precoding matrix to transform $\left\lvert\Lambda_{m,k}\right\rvert$-dimensional vector $\hat{\mathbf{g}}_{\Lambda_{m,k}}$ into $N$-dimensional vector $\mathbf{w}_{m,k}$ and $\hat{\mathbf{g}}_{\Lambda_{m,k}}$ is the dominating PGI vector fed back from the user. In the following theorem, we express the achievable rate of the dominating PGI feedback scheme as a function of the dominating path indices $\lbrace\Lambda_{m,k}\rbrace$ and the precoding matrices $\lbrace\mathbf{V}_{\Lambda_{m,k}}\rbrace$. Based on this, we can find $\lbrace \Lambda_{m,k}\rbrace$ and $\lbrace \mathbf{V}_{\Lambda_{m,k}}\rbrace$ maximizing the sum rate performance of the dominating PGI feedback.
\begin{theorem}
The achievable rate $R_{k}^{(\textup{ideal})}$ of the user $k$ for the ideal system with perfect PGI is
\begin{align}\label{3.2.2}
R_{k}^{(\textup{ideal})}\left(\lbrace\Lambda_{m,k}\rbrace,\lbrace\mathbf{V}_{\Lambda_{m,k}}\rbrace\right)=\log_{2}\left(1+\frac{\left\lvert\sum_{m=1}^{M}\textup{tr}\left(\mathbf{A}_{\Lambda_{m,k}}^{\textrm{H}}\mathbf{V}_{\Lambda_{m,k}}\right)\right\rvert^{2}+\sum_{m=1}^{M}\left\lVert\mathbf{A}_{m,k}^{\textrm{H}}\mathbf{V}_{\Lambda_{m,k}}\right\rVert_{\textup{F}}^{2}}{\sum_{j\neq k}^{K}\sum_{m=1}^{M}\left\lVert\mathbf{A}_{m,k}^{\textrm{H}}\mathbf{V}_{\Lambda_{m,j}}\right\rVert_{\textup{F}}^{2}+\sigma_{n}^{2}}\right)
\end{align}
and the corresponding sum rate is $R_{\text{tot}}=\sum_{k=1}^{K}R_{k}^{(\textup{ideal})}$ where $\mathbf{A}_{m,k}\in\mathbb{C}^{N\times P}$ is the array steering matrix in \eqref{2.1.3} and $\mathbf{A}_{\Lambda_{m,k}}=[\mathbf{a}(\theta_{m,k,i}),\,i\in\Lambda_{m,k}]\in\mathbb{C}^{N\times\lvert\Lambda_{m,k}\rvert}$ is the submatrix of $\mathbf{A}_{m,k}$.
\end{theorem}
\begin{proof}
See Appendix A.
\end{proof}
\noindent Then, the dominating path selection problem to choose $L$ paths maximizing the sum rate for each user can be formulated as
\begin{subequations}\label{3.2.3}
\begin{align}
\mathcal{P}_{1}:\underset{\lbrace \Lambda_{m,k}\rbrace,\lbrace\mathbf{V}_{\Lambda_{m,k}}\rbrace}{\text{max}}&\,R_{\text{tot}}\left(\lbrace\Lambda_{m,k}\rbrace,\lbrace\mathbf{V}_{\Lambda_{m,k}}\rbrace\right)\label{3.2.3.1}\\
\text{s.t.}\,\quad\quad &\sum_{m=1}^{M}\lvert\Lambda_{m,k}\rvert=L,\quad \forall k\in\mathcal{U}\label{3.2.3.2}\\
&\left\lVert\mathbf{V}_{\Lambda_{m,k}}\right\rVert_{\text{F}}=1,\quad \forall m\in\mathcal{B},\,\forall k\in\mathcal{U}.\label{3.2.3.3}
\end{align}
\end{subequations}
Note that \eqref{3.2.3.2} is the dominating path number constraint and \eqref{3.2.3.3} is the transmit power constraint. 

\subsection{Alternating Dominating Path Selection and Precoding Algorithm}
Major obstacle in solving $\mathcal{P}_{1}$ is the strong correlation between the dominating path index set $\Lambda_{m,k}$ and the precoding matrix $\mathbf{V}_{\Lambda_{m,k}}$. In fact, since the column dimension of $\mathbf{V}_{\Lambda_{m,k}}$ is the number of dominating paths $\lvert \Lambda_{m,k}\rvert$, $\Lambda_{m,k}$ and $\mathbf{V}_{\Lambda_{m,k}}$ cannot be determined simultaneously. In this subsection, we propose an algorithm to determine $\lbrace\Lambda_{m,k}\rbrace$ and $\lbrace\mathbf{V}_{\Lambda_{m,k}}\rbrace$ in an alternating way: 1) First, we fixed $\lbrace\Lambda_{m,k}\rbrace$ and find out the optimal precoding matrices $\lbrace\mathbf{V}_{\Lambda_{m,k}}\rbrace$ maximizing the sum rate. 2) We then update $\lbrace\Lambda_{m,k}\rbrace$ by removing the path index giving the minimal impact on the sum rate. We repeat these procedures until $L$ dominating paths remain for each user.

\subsubsection{Precoding Matrix Optimization}
We first discuss the way to find out the optimal precoding matrices $\lbrace\mathbf{V}_{\Lambda_{m,k}}\rbrace$ when $\lbrace\Lambda_{m,k}\rbrace$ are fixed. Unfortunately, this problem shown in \eqref{3.3.1} is highly non-convex and also contains multiple matrix variables. To address these issues, we first vectorize and concatenate the variables of multiple BSs $\mathbf{V}_{\Lambda_{1,k}},\cdots,\mathbf{V}_{\Lambda_{M,k}}$ into $\mathbf{x}_{\Lambda_{k}}$. Then, by exploiting the notion of \textit{leakage}, we decompose the sum rate maximization problem into distributed leakage minimization problems for each $\mathbf{x}_{\Lambda_{k}}$ to obtain the tractable closed-form solution. Finally, we de-vectorize and de-concatenate $\mathbf{x}_{\Lambda_{k}}$ to obtain the desired precoding matrices $\mathbf{V}_{\Lambda_{1,k}},\cdots,\mathbf{V}_{\Lambda_{M,k}}$. 

By plugging \eqref{3.2.2} into $\mathcal{P}_{1}$, we obtain
\begin{subequations}\label{3.3.1}
\begin{align}
\mathcal{P}_{2}:\underset{\lbrace\mathbf{V}_{\Lambda_{m,k}}\rbrace,\lbrace t_{k}\rbrace}{\text{max}} &\sum_{k=1}^{K}t_{k}\label{3.3.1.1}\\
\text{s.t.}\,\,\,\,\quad &\frac{\left\lvert\sum_{m=1}^{M}\textup{tr}\left(\mathbf{A}_{\Lambda_{m,k}}^{\textrm{H}}\mathbf{V}_{\Lambda_{m,k}}\right)\right\rvert^{2}+\sum_{m=1}^{M}\left\lVert\mathbf{A}_{m,k}^{\textrm{H}}\mathbf{V}_{\Lambda_{m,k}}\right\rVert_{\textup{F}}^{2}}{\sum_{j\neq k}^{K}\sum_{m=1}^{M}\left\lVert\mathbf{A}_{m,k}^{\textrm{H}}\mathbf{V}_{\Lambda_{m,j}}\right\rVert_{\textup{F}}^{2}+\sigma_{n}^{2}}\geq 2^{t_{k}}-1,\,\,\,\, \forall k\in\mathcal{U} \label{3.3.1.2}\\
&\left\lVert\mathbf{V}_{\Lambda_{m,k}}\right\rVert_{\text{F}}=1,\quad \forall m\in\mathcal{B},\,\forall k\in\mathcal{U},\label{3.3.1.3}
\end{align}
\end{subequations}
where $\lbrace t_{k}\rbrace$ are the auxiliary variables. Then, we vectorize the optimization variables ($\mathbf{x}_{\Lambda_{m,k}}=\text{vec}\left(\mathbf{V}_{\Lambda_{m,k}}\right)$, $\boldsymbol{\mu}_{\Lambda_{m,k}}=\text{vec}\left(\mathbf{A}_{\Lambda_{m,k}}\right)$)
and then concatenate the variables of multiple BSs  ($\mathbf{x}_{\Lambda_{k}}=\left[\mathbf{x}_{\Lambda_{1,k}}^{\textrm{T}},\,\cdots,\,\mathbf{x}_{\Lambda_{M,k}}^{\textrm{T}}\right]^{\textrm{T}}$, $\boldsymbol{\mu}_{\Lambda_{k}}=\left[\boldsymbol{\mu}_{\Lambda_{1,k}}^{\textrm{T}},\,\cdots,\,\boldsymbol{\mu}_{\Lambda_{M,k}}^{\textrm{T}}\right]^{\textrm{T}}$) to obtain 
\begin{subequations}\label{3.3.2}
\begin{align}
\mathcal{P}_{3}:\,\underset{\lbrace\mathbf{x}_{\Lambda_{k}}\rbrace,\,\lbrace t_{k}\rbrace}{\text{max}}\,\, &\sum_{k=1}^{K}t_{k}\label{3.3.2.1}\\
\text{s.t.}\,\,\,\,\,\quad &\frac{\left\lvert\boldsymbol{\mu}_{\Lambda_{k}}^{\textrm{H}}\mathbf{x}_{\Lambda_{k}}\right\rvert^{2}+\mathbf{x}_{\Lambda_{k}}^{\textrm{H}}\boldsymbol{\Gamma}_{k,k}\mathbf{x}_{\Lambda_{k}}}{\sum_{j\neq k}^{K}\mathbf{x}_{\Lambda_{j}}^{\textrm{H}}\boldsymbol{\Gamma}_{j,k}\mathbf{x}_{\Lambda_{j}}+\sigma_{n}^{2}}\geq 2^{t_{k}}-1,\quad \forall k\in\mathcal{U} \label{3.3.2.2}\\
&\left\lVert\mathbf{x}_{\Lambda_{k}}\right\rVert=\sqrt{M},\quad \forall m\in\mathcal{B},\,\forall k\in\mathcal{U},\label{3.3.2.3}
\end{align}
\end{subequations}
where $\boldsymbol{\Gamma}_{j,k}=\text{diag}\left(\mathbf{I}_{\lvert\Lambda_{1,j}\rvert}\otimes\mathbf{A}_{1,k}\mathbf{A}_{1,k}^{\textrm{H}},\,\cdots,\,\mathbf{I}_{\lvert\Lambda_{M,j}\rvert}\otimes\mathbf{A}_{M,k}\mathbf{A}_{M,k}^{\textrm{H}}\right)$. Here, we use the properties $\text{tr}\left(\!\mathbf{A}_{\Lambda_{m,k}}^{\textrm{H}}\!\mathbf{V}_{\Lambda_{m,k}}\!\right)\!=\!\text{vec}\left(\mathbf{A}_{\Lambda_{m,k}}\right)^{\textrm{H}}\!\text{vec}\left(\mathbf{V}_{\Lambda_{m,k}}\right)$ and $\lVert\mathbf{A}_{m,k}^{\textrm{H}}\mathbf{V}_{\Lambda_{m,j}}\rVert_{\text{F}}\!=\!\left\lVert\left(\mathbf{I}_{\lvert\Lambda_{m,j}\rvert}\!\otimes\!\mathbf{A}_{m,k}^{\textrm{H}}\right)\text{vec}\left(\mathbf{V}_{\Lambda_{m,j}}\right)\right\rVert$. Also, since it is difficult to satisfy the norm constraint $\left\lVert\mathbf{V}_{\Lambda_{m,k}}\right\rVert_{\text{F}}=\left\lVert\mathbf{x}_{\Lambda_{m,k}}\right\rVert=1$ for each and every $m\in\mathcal{B}$, we use a relaxed normalized constraint $\left\lVert\mathbf{x}_{\Lambda_{k}}\right\rVert=\sqrt{\sum_{m=1}^{M}\lVert\mathbf{x}_{\Lambda_{m,k}}\rVert^{2}}=\sqrt{M}$ in $\mathcal{P}_{3}$.

The modified problem $\mathcal{P}_{3}$ looks simpler than the original problem $\mathcal{P}_{2}$, but it is still hard to find the optimal solution. The reason is because the rate constraint \eqref{3.3.2.2} is a non-convex quadratic fractional function (i.e., both numerator and denominator are quadratic functions) so that $\mathcal{P}_{3}$ is a non-convex optimization problem. Further, $\mathcal{P}_{3}$ requires joint optimization for $\mathbf{x}_{\Lambda_{1,k}},\cdots,\mathbf{x}_{\Lambda_{M,k}}$, and thus it is very difficult to find out the global solutions simultaneously. As a remedy, we introduce the notion of \textit{leakage}, a measure of how much signal power leaks into the other users~\cite{sadek2007leakage}. To be specific, the signal-to-leakage-and-noise-ratio (SLNR) of the user $k$ is given by 
\begin{align}\label{3.3.3}
\text{SLNR}_{k}=\frac{\mathbb{E}\left[\left\lvert\sum_{m=1}^{M}\mathbf{h}_{m,k}^{\textrm{H}}\mathbf{w}_{m,k}\right\rvert^{2}\right]}{\sum_{j\neq k}^{K}\mathbb{E}\left[\left\lvert\sum_{m=1}^{M}\mathbf{h}_{m,j}^{\textrm{H}}\mathbf{w}_{m,k}\right\rvert^{2}\right]+\sigma_{n}^{2}}\stackrel{(a)}{=}\frac{\left\lvert\boldsymbol{\mu}_{k}^{\textrm{H}}\mathbf{x}_{\Lambda_{k}}\right\rvert^{2}+\mathbf{x}_{\Lambda_{k}}^{\textrm{H}}\boldsymbol{\Gamma}_{k,k}\mathbf{x}_{\Lambda_{k}}}{\sum_{j\neq k}^{K}\mathbf{x}_{\Lambda_{k}}^{\textrm{H}}\boldsymbol{\Gamma}_{k,j}\mathbf{x}_{\Lambda_{k}}+\sigma_{n}^{2}}.
\end{align}
where $(a)$ comes from \eqref{3.3.2.2}\footnote{When compared to the signal-to-interference-and-noise-ratio (SINR) of the user $k$ in \eqref{2.1.5}, one can observe that the only difference is the exchange of user index at the denominator between $\mathbf{h}_{m,j}^{\textrm{H}}\mathbf{w}_{m,k}$ and $\mathbf{h}_{m,k}^{\textrm{H}}\mathbf{w}_{m,j}$. Hence, we can easily obtain the closed-form expression of $\text{SLNR}_{k}$ from \eqref{3.3.2.2}.}. Note that while \eqref{3.3.2.2} is a function of $\mathbf{x}_{\Lambda_{1}},\cdots,\mathbf{x}_{\Lambda_{K}}$, $\text{SLNR}_{k}$ in \eqref{3.3.3} is a sole function of $\mathbf{x}_{\Lambda_{k}}$. Thus, for each user $k$, we can find out the optimal $\mathbf{x}_{\Lambda_{k}}^{*}$ maximizing $\text{SLNR}_{k}$ separately. While the solution is a bit sub-optimal, it is simple and easy to calculate~\cite{sadek2007leakage}.

The distributed SLNR maximization problem for the user $k$ is given by 
\begin{align}\label{3.3.4}
\mathcal{P}_{4}:\mathbf{x}_{\Lambda_{k}}^{*}=\text{arg}\underset{\left\lVert\mathbf{x}_{\Lambda_{k}}\right\rVert=\sqrt{M}}{\text{max}}\,\,&\frac{\left\lvert\boldsymbol{\mu}_{\Lambda_{k}}^{\textrm{H}}\mathbf{x}_{\Lambda_{k}}\right\rvert^{2}+\mathbf{x}_{\Lambda_{k}}^{\textrm{H}}\boldsymbol{\Gamma}_{k,k}\mathbf{x}_{\Lambda_{k}}}{\sum_{j\neq k}^{K}\mathbf{x}_{\Lambda_{k}}^{\textrm{H}}\boldsymbol{\Gamma}_{k,j}\mathbf{x}_{\Lambda_{k}}+\sigma_{n}^{2}},\quad\forall k\in\mathcal{U}.
\end{align} 
Using the normalization constraint, we can simplify the objective function of $\mathcal{P}_{4}$ as
\begin{align}\label{3.3.5}
\frac{\left\lvert\boldsymbol{\mu}_{\Lambda_{k}}^{\textrm{H}}\mathbf{x}_{\Lambda_{k}}\right\rvert^{2}+\mathbf{x}_{\Lambda_{k}}^{\textrm{H}}\boldsymbol{\Gamma}_{k,k}\mathbf{x}_{\Lambda_{k}}}{\sum_{j\neq k}^{K}\mathbf{x}_{\Lambda_{k}}^{\textrm{H}}\boldsymbol{\Gamma}_{k,j}\mathbf{x}_{\Lambda_{k}}+\sigma_{n}^{2}}&=\frac{\mathbf{x}_{\Lambda_{k}}^{\textrm{H}}\left(\boldsymbol{\mu}_{\Lambda_{k}}\boldsymbol{\mu}_{\Lambda_{k}}^{\textrm{H}}+\boldsymbol{\Gamma}_{k,k}\right)\mathbf{x}_{\Lambda_{k}}}{\mathbf{x}_{\Lambda_{k}}^{\textrm{H}}\left(\sum_{j\neq k}^{K}\boldsymbol{\Gamma}_{k,j}+\frac{\sigma_{n}^{2}}{M}\mathbf{I}_{N\left\lvert\Lambda_{k}\right\rvert}\right)\mathbf{x}_{\Lambda_{k}}}\\
&=\frac{\mathbf{x}_{\Lambda_{k}}^{\textrm{H}}\mathbf{U}_{k}\mathbf{x}_{\Lambda_{k}}}{\mathbf{x}_{\Lambda_{k}}^{\textrm{H}}\mathbf{W}_{k}\mathbf{x}_{\Lambda_{k}}},
\end{align}
where $\mathbf{U}_{k}=\boldsymbol{\mu}_{\Lambda_{k}}\boldsymbol{\mu}_{\Lambda_{k}}^{\textrm{H}}+\boldsymbol{\Gamma}_{k,k}$ and $\mathbf{W}_{k}=\sum_{j\neq k}^{K}\boldsymbol{\Gamma}_{k,j}+\frac{\sigma_{n}^{2}}{M}\mathbf{I}_{N\left\lvert\Lambda_{k}\right\rvert}$. Then, $\mathcal{P}_{4}$ can be re-expressed as
\begin{align}\label{3.3.6}
\mathcal{P}_{4}:\,\mathbf{x}_{\Lambda_{k}}^{*}=\text{arg}\underset{\left\lVert\mathbf{x}_{\Lambda_{k}}\right\rVert=\sqrt{M}}{\text{max}}\,\,&\frac{\mathbf{x}_{\Lambda_{k}}^{\textrm{H}}\mathbf{U}_{k}\mathbf{x}_{\Lambda_{k}}}{\mathbf{x}_{\Lambda_{k}}^{\textrm{H}}\mathbf{W}_{k}\mathbf{x}_{\Lambda_{k}}},\quad \forall k\in\mathcal{U}.
\end{align}
\begin{lemma}
The solution $\mathbf{x}_{\Lambda_{k}}^{*}$ of $\mathcal{P}_{4}$ is given by~\cite{sadek2007leakage}
\begin{align}\label{3.3.7}
\mathbf{x}_{\Lambda_{k}}^{*}=\sqrt{M}\frac{\mathbf{u}_{k,\text{max}}}{\lVert\mathbf{u}_{k,\text{max}}\rVert},
\end{align}
where $\mathbf{u}_{k,\text{max}}$ is the eigenvector corresponding to the largest eigenvalue of $\mathbf{W}_{k}^{-1}\mathbf{U}_{k}$.
\end{lemma}
\noindent Using Lemma 1, we can easily obtain the closed-form solution $\mathbf{x}_{\Lambda_{k}}^{*}$ of $\mathcal{P}_{4}$. From the de-vectorization and de-concatenation of $\mathbf{x}_{\Lambda_{k}}^{*}$, we obtain the desired matrices $\mathbf{V}_{\Lambda_{1,k}}^{*},\cdots,\mathbf{V}_{\Lambda_{M,k}}^{*}$. 

\subsubsection{Dominating Path Index Update}
Once we obtain $\lbrace\mathbf{V}_{\Lambda_{m,k}}\rbrace$ from the precoding matrix optimization, we then update the dominating path indices $\lbrace\Lambda_{m,k}\rbrace$ by removing the path index giving the minimal impact on the sum rate. In particular, for each user $k$, we choose the path index $\hat{i}_{k}$ corresponding to the minimum $l_{2}$-norm column vector of $\left[\mathbf{V}_{\Lambda_{1,k}},\cdots,\mathbf{V}_{\Lambda_{M,k}}\right]$ as
\begin{align}
(\hat{m}_{k},\,\hat{i}_{k})=\text{arg}\underset{m\in\mathcal{B},\,i\in\Lambda_{m,k}}{\text{min}}\left\lVert\mathbf{v}_{m,k,i}\right\rVert,
\end{align}
and then remove $\hat{i}_{k}$ from $\Lambda_{\hat{m}_{k},k}$. Note that $\mathbf{v}_{m,k,i}$ is the column vector of $\mathbf{V}_{\Lambda_{m,k}}$ corresponding to the $i$-th path from the BS $m$ to the user $k$. The intuition behind this choice is because 
\begin{align}
\mathbb{E}\left[\left\lVert\mathbf{w}_{m,k}\right\rVert^{2}\right]&=\mathbb{E}\left[\left\lVert\sum_{i\in\Lambda_{m,k}}\hat{g}_{m,k,i}\mathbf{v}_{m,k,i}\right\rVert^{2}\right]\\
&=\sum_{i\in\Lambda_{m,k}}\lVert\mathbf{v}_{m,k,i}\rVert^{2}\mathbb{E}\left[\lvert \hat{g}_{m,k,i}\rvert^{2}\right]\\
&=\sum_{i\in\Lambda_{m,k}}\lVert\mathbf{v}_{m,k,i}\rVert^{2},
\end{align}
and thus, the removal of the minimum $l_{2}$-norm column vector $\mathbf{v}_{\hat{m}_{k},k,\hat{i}_{k}}$ would give a minimal impact on $\mathbf{w}_{m,k}$. In addition, since the sum rate is a function of $\mathbf{w}_{m,k}$, it is quite reasonable to assume that the removal of corresponding path index $\hat{i}_{k}$ would also give a minimal impact on the sum rate\footnote{Even though $L$ is chosen to be larger than the effective number of propagation paths, the precoding matrix would be optimized such that the transmit power is focused on the best column vectors (corresponding to the dominant paths).}. The precoding matrix optimization and the dominating path index update are repeated iteratively until only $L$ paths remain for each user. The proposed alternating algorithm is summarized in Table I.

Once the dominating paths maximizing the sum rate are chosen, each user acquires the corresponding dominating PGI from the downlink pilot signal, quantizes the acquired dominating PGI, and then feeds it back to the BSs. In the following section, we will discuss this issue in detail. 

\begin{algorithm}[t]
\begin{algorithmic}[0]
\item[\textbf{Input:}] Path AoDs $\lbrace\theta_{m,k,i}\rbrace$, BS set  $\,\mathcal{B}$, user set $\,\mathcal{U}$, number of propagation paths $P$,
\Statex \quad \quad \,\,number of dominating paths $L$ 
\item[\textbf{Initialization:}] $\Lambda_{m,k}=\lbrace 1,\cdots,P\rbrace,\quad \forall m\in\mathcal{B},\,\forall k\in\mathcal{U}$
\Statex \quad\quad\quad\quad\quad $\lbrace\mathbf{V}_{\Lambda_{m,k}}\rbrace=\text{Precoding\_matrix\_optimization}\left(\lbrace\theta_{m,k,i}\rbrace,\,\lbrace\Lambda_{m,k}\rbrace\right)$
\algrule
\item[\textbf{Iteration:}]
\While{$\sum_{m=1}^{M}\lvert\Lambda_{m,k}\rvert>L$ for some $k$} \Comment{Check the number of dominating paths}
\For{$k\in\mathcal{U}$}
\If{$\sum_{m=1}^{M}\lvert\Lambda_{m,k}\rvert>L$}
\State $(\hat{m}_{k},\hat{i}_{k})=\text{arg}\underset{m\in\mathcal{B},\,i\in\Lambda_{m,k}}{\text{min}}\left\lVert\mathbf{v}_{m,k,i}\right\rVert$\Comment{Find the minimal $l_{2}$-norm column vector}
\State $\Lambda_{\hat{m}_{k},k}=\Lambda_{\hat{m}_{k},k}\setminus \lbrace \hat{i}_{k}\rbrace$\Comment{Remove the corresponding path index}
\EndIf
\EndFor
\State $\lbrace\mathbf{V}_{\Lambda_{m,k}}\rbrace=\text{Precoding\_matrix\_optimization}\left(\lbrace\theta_{m,k,i}\rbrace,\,\lbrace\Lambda_{m,k}\rbrace\right)$
\EndWhile
\algrule
\item[\textbf{Function}] $\text{Precoding\_matrix\_optimization}\left(\lbrace\theta_{m,k,i}\rbrace,\,\lbrace\Lambda_{m,k}\rbrace\right)$
\State $\boldsymbol{\mu}_{\Lambda_{m,k}}=\text{vec}\left(\mathbf{A}_{\Lambda_{m,k}}\right), \boldsymbol{\mu}_{\Lambda_{k}}=\left[\boldsymbol{\mu}_{\Lambda_{1,k}}^{\textrm{T}},\,\cdots,\,\boldsymbol{\mu}_{\Lambda_{M,k}}^{\textrm{T}}\right]^{\textrm{T}},\quad \forall m\in\mathcal{B},\,\forall k\in\mathcal{U}$
\State $\boldsymbol{\Gamma}_{j,k}=\text{diag}\left(\mathbf{I}_{\left\lvert\Lambda_{1,j}\right\rvert}\otimes\mathbf{A}_{1,k}\mathbf{A}_{1,k}^{\textrm{H}},\,\cdots,\,\mathbf{I}_{\left\lvert\Lambda_{M,j}\right\rvert}\otimes\mathbf{A}_{M,k}\mathbf{A}_{M,k}^{\textrm{H}}\right),\quad \forall j,k\in\mathcal{U}$
\For{$k\in\mathcal{U}$}
\State $\mathbf{U}_{k}=\boldsymbol{\mu}_{\Lambda_{k}}\boldsymbol{\mu}_{\Lambda_{k}}^{\textrm{H}}+\boldsymbol{\Gamma}_{k,k},\,\mathbf{W}_{k}=\sum_{j\neq k}^{K}\boldsymbol{\Gamma}_{k,j}+\frac{\sigma_{n}^{2}}{M}\mathbf{I}_{N\left\lvert\Lambda_{k}\right\rvert}$
\State $\mathbf{u}_{k,\text{max}}=\text{max\_eigenvector}\left(\mathbf{W}_{k}^{-1}\mathbf{U}_{k}\right)$
\State $\hat{\mathbf{x}}_{\Lambda_{k}}=\sqrt{M}\frac{\mathbf{u}_{k,\text{max}}}{\lVert\mathbf{u}_{k,\text{max}}\rVert}$
\State $\left[\hat{\mathbf{x}}_{\Lambda_{1,k}}^{\textrm{T}},\cdots,\hat{\mathbf{x}}_{\Lambda_{M,k}}^{\textrm{T}}\right]^{\textrm{T}}=\hat{\mathbf{x}}_{\Lambda_{k}}$
\State $\hat{\mathbf{V}}_{\Lambda_{m,k}}=\text{vec}^{-1}\left(\hat{\mathbf{x}}_{\Lambda_{m,k}}\right),\quad \forall m\in\mathcal{B}$
\EndFor
\State \textbf{return} $\lbrace\hat{\mathbf{V}}_{\Lambda_{m,k}}\rbrace$
\item[\textbf{end function}] $\,$
\algrule
\item[\textbf{Output:}] $\lbrace \Lambda_{m,k}\rbrace$, $\lbrace\mathbf{V}_{\Lambda_{m,k}}\rbrace$
\captionof{table}{Alternating dominating path selection and precoding algorithm}
\end{algorithmic}
\end{algorithm}
\clearpage 

\begin{figure}[t]
\centering
\includegraphics[scale=0.43]{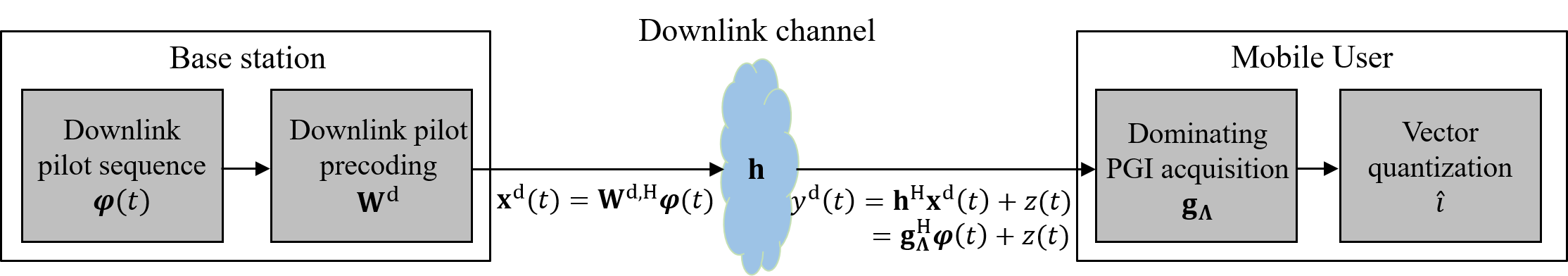}
\caption{Downlink pilot precoding for dominating PGI acquisition}
\end{figure}

\section{Downlink Pilot Precoding for Dominating Path Gain Information Acquisition}
In the FDD systems, a user acquires the channel information from the downlink pilot signal and then feeds the quantized channel vector back to the BS. In contrast, in the proposed scheme, a user acquires the dominating PGI and then feeds back the quantized value to BS. There are however some difficulties in the dominating PGI acquisition. First, since each user needs to selectively feed back PGIs of the dominating paths, the BS must assign additional resources to indicate the desired path information. Also, it is computationally inefficient for the user to estimate the gain of all possible paths. To handle this issue, we propose a new downlink training scheme using spatially precoded pilot signal in the acquisition of dominating PGI. 

In essence, the goal of precoded pilot signal is to convert the downlink channel vector into the dominating PGI vector so that the user can easily estimate the dominating PGI using the conventional channel estimation techniques such as the linear minimum mean square error (LMMSE) estimator~\cite{choi2014downlink} (see Fig. 5). Additionally, since the dimension of dominating PGI (i.e., the number of dominating paths) is reduced and thus becomes much smaller than that of the downlink CSI (i.e., the number of transmit antennas), we can achieve a reduction in the pilot resources. 

When the pilot precoding matrix $\mathbf{W}_{m,k}^{\textrm{d}}\in\mathbb{C}^{\lvert\Lambda_{m,k}\rvert\times N}$ is applied, the downlink precoded pilot signal $\mathbf{x}_{m}^{\textrm{d}}(t)\in\mathbb{C}^{N}$ of the BS $m$ at time slot $t$ is given by
\begin{align}\label{4.1.1}
\mathbf{x}_{m}^{\textrm{d}}(t)=\sum_{k=1}^{K}\mathbf{W}_{m,k}^{\textrm{d},\textrm{H}}\boldsymbol{\psi}_{m,k}(t),\quad t=1,\cdots,\tau
\end{align}
where $\lbrace\boldsymbol{\psi}_{m,k}(t)\rbrace_{t=1}^{\tau}\subseteq\mathbb{C}^{\lvert\Lambda_{m,k}\rvert}$ is the downlink pilot sequence from the BS $m$ to the user $k$. Then, the received signal $y_{k}^{\textrm{d}}(t)\in\mathbb{C}$ of the user $k$ at time slot $t$ is 
\begin{align}\label{4.1.2}
y_{k}^{\textrm{d}}(t)=\sum_{m=1}^{M}\left(\mathbf{W}_{m,k}^{\textrm{d}}\mathbf{h}_{m,k}\right)^{\textrm{H}}\boldsymbol{\psi}_{m,k}(t)+\sum_{m=1}^{M}\sum_{j\neq k}^{K}\left(\mathbf{W}_{m,j}^{\textrm{d}}\mathbf{h}_{m,k}\right)^{\textrm{H}}\boldsymbol{\psi}_{m,j}(t)+z_{k}(t)
\end{align}
where $z_{k}(t)\sim\mathcal{CN}(0,\sigma_{z}^{2})$ is the Gaussian noise. The user $k$ collects this received signal for each slot, i.e., $\mathbf{y}_{k}^{\textrm{d}}=\left[y_{k}^{\textrm{d}}(1),\cdots,y_{k}^{\textrm{d}}(\tau)\right]^{\textrm{H}}$ and then multiplies $\boldsymbol{\Psi}_{m,k}=\left[\boldsymbol{\psi}_{m,k}(1),\cdots,\boldsymbol{\psi}_{m,k}(\tau)\right]$ to get
\begin{align}
\boldsymbol{\Psi}_{m,k}\mathbf{y}_{k}^{\textrm{d}}&=\boldsymbol{\Psi}_{m,k}\left(\sum_{m=1}^{M}\boldsymbol{\Psi}_{m,k}^{\textrm{H}}\mathbf{W}_{m,k}^{\textrm{d}}\mathbf{h}_{m,k}+\sum_{m=1}^{M}\sum_{j\neq k}^{K}\boldsymbol{\Psi}_{m,j}^{\textrm{H}}\mathbf{W}_{m,j}^{\textrm{d}}\mathbf{h}_{m,k}+\mathbf{z}_{k}\right)\label{4.1.3.1} \\
&\stackrel{(a)}{=}\mathbf{W}_{m,k}^{\textrm{d}}\mathbf{h}_{m,k}+\mathbf{n}_{k}\label{4.1.3.2}
\end{align}
where $\mathbf{z}_{k}=\left[z_{k}(1),\cdots,z_{k}(\tau)\right]^{\textrm{H}}$ and $\mathbf{n}_{k}=\boldsymbol{\Psi}_{m,k}\mathbf{z}_{k}$. Also, $(a)$ is due to the orthogonality of pilot sequence. 

From \eqref{4.1.3.2}, we observe that if the BS uses a precoding matrix $\mathbf{W}_{m,k}^{\textrm{d}}$ satisfying $\mathbf{W}_{m,k}^{\textrm{d}}\mathbf{h}_{m,k}=\mathbf{g}_{\Lambda_{m,k}}$, then one can extract the dominating PGI vector $\mathbf{g}_{\Lambda_{m,k}}$ from $\boldsymbol{\Psi}_{m,k}\mathbf{y}_{k}^{\textrm{d}}$. To generate the desired precoding matrix $\mathbf{W}_{m,k}^{\textrm{d}}$, we basically need to perform two operations: 1) application of the matrix inversion of $\mathbf{A}_{m,k}^{+}=\left(\mathbf{A}_{m,k}^{\textrm{H}}\mathbf{A}_{m,k}\right)^{-1}\mathbf{A}_{m,k}^{\textrm{H}}$ and 2) compression of $\mathbf{g}_{m,k}$ into $\mathbf{g}_{\Lambda_{m,k}}$. Note that $\mathbf{A}_{m,k}^{+}$ exists as long as $\mathbf{A}_{m,k}^{\textrm{H}}\mathbf{A}_{m,k}$ is invertible, which is easily guaranteed by the fact that the array steering vectors corresponding to different AoDs are independent and the number of transmit antennas $N$ is larger then the number of paths $P$. Thus,
\begin{align}\label{4.1.4}
\mathbf{A}_{m,k}^{+}\mathbf{h}_{m,k}\stackrel{(a)}{=}\mathbf{A}_{m,k}^{+}\mathbf{A}_{m,k}\mathbf{g}_{m,k}=\mathbf{g}_{m,k}
\end{align}
where $(a)$ is from \eqref{2.1.3}. Once $\mathbf{g}_{m,k}$ is obtained, we then extract $\mathbf{g}_{\Lambda_{m,k}}$ from $\mathbf{g}_{m,k}$ using the path selection matrix $\mathbf{G}_{m,k}$. For example, if the number of propagation paths is $3$ and $\Lambda_{m,k}=\lbrace 1,\,3\rbrace$, then $\mathbf{G}_{m,k}=
\small{
\begin{bmatrix}
1 & 0 & 0 \\
0 & 0 & 1 \\
\end{bmatrix}}$
and thus,  
\begin{align}\label{4.1.5}
\mathbf{G}_{m,k}\mathbf{g}_{m,k}=\begin{bmatrix}
1 & 0 & 0 \\
0 & 0 & 1 \\
\end{bmatrix}
\begin{bmatrix}
g_{m,k,1}\\
g_{m,k,2}\\
g_{m,k,3}\\
\end{bmatrix}
=
\begin{bmatrix}
g_{m,k,1}\\
g_{m,k,3}\\
\end{bmatrix}
=\mathbf{g}_{\Lambda_{m,k}}
\end{align} 
In summary, the pilot precoding matrix $\mathbf{W}_{m,k}^{\textrm{d}}$ from the BS $m$ to the user $k$ is given by
\begin{align}\label{4.1.6}
\mathbf{W}_{m,k}^{\textrm{d}}=\mathbf{G}_{m,k}\mathbf{A}_{m,k}^{+}
\end{align} 
Using $\mathbf{W}_{m,k}^{\textrm{d}}$ in \eqref{4.1.6}, we can convert $\mathbf{h}_{m,k}$ into $\mathbf{g}_{\Lambda_{m,k}}$ (i.e., $\mathbf{W}_{m,k}^{\textrm{d}}\mathbf{h}_{m,k}=\mathbf{g}_{\Lambda_{m,k}}$). Hence, \eqref{4.1.3.2} can be re-expressed as
\begin{align}\label{4.1.7}
\boldsymbol{\Psi}_{m,k}\mathbf{y}_{k}^{\textrm{d}}=\mathbf{g}_{\Lambda_{m,k}}+\mathbf{n}_{k},
\end{align}
Finally, the user $k$ acquires $\hat{\mathbf{g}}_{\Lambda_{m,k}}$ from $\boldsymbol{\Psi}_{m,k}\mathbf{y}_{k}^{\textrm{d}}$ by using the linear MMSE estimation~\cite{choi2014downlink} as
\begin{align}\label{4.1.8}
\hat{\mathbf{g}}_{\Lambda_{m,k}}=\frac{1}{1+\sigma_{z}^{2}}\boldsymbol{\Psi}_{m,k}\mathbf{y}_{k}^{\textrm{d}}
\end{align}

After the estimation of the dominating PGI, each user quantizes it and then feeds back to the BS. To be specific, the user $k$ concatenates $\mathbf{g}_{\Lambda_{1,k}},\cdots,\mathbf{g}_{\Lambda_{M,k}}$ into a single vector $\mathbf{g}_{\Lambda_{k}}=\left[\mathbf{g}_{\Lambda_{1,k}}^{\textrm{T}},\cdots,\mathbf{g}_{\Lambda_{M,k}}^{\textrm{T}}\right]^{\textrm{T}}\in\mathbb{C}^{L}$ and then quantizes $\mathbf{g}_{\Lambda_{k}}$ into a codeword index $\hat{i}_{k}$ as
\begin{align}\label{4.1.9}
\hat{i}_{k}=\text{arg}\,\underset{i}{\text{max}}\left\lvert\bar{\mathbf{g}}_{\Lambda_{k}}^{\textrm{H}}\mathbf{c}_{i}\right\rvert^{2}
\end{align}
where $\bar{\mathbf{g}}_{\Lambda_{k}}=\mathbf{g}_{\Lambda_{m,k}}/\left\lVert\mathbf{g}_{\Lambda_{m,k}}\right\rVert$ and $\mathbf{c}_{i}$ is the codeword. For example, one can use the random vector quantization (RVQ) codebook~\cite{jindal2006mimo}. After receiving $\hat{i}_{k}$, DU reconstructs the original dominating PGI as $\hat{\mathbf{g}}_{\Lambda_{k}}=\left\lVert\mathbf{g}_{\Lambda_{k}}\right\rVert\mathbf{c}_{\hat{i}_{k}}$ where $\left\lVert\mathbf{g}_{\Lambda_{k}}\right\rVert$ is the channel magnitude feedback of user.
\\

\section{Performance Analysis of the Proposed Dominating Path Gain Information Feedback}
In this section, we study the performance of the proposed dominating PGI feedback scheme. We first analyze the distortion induced from the quantization of dominating PGI vector $\mathbf{g}_{\Lambda_{k}}$ and then analyze the rate gap between the ideal system with perfect PGI and the realistic system with finite rate PGI feedback. Finally, we compute the number of feedback bits required to maintain a constant rate gap with the ideal system. 

\subsection{Quantization Distortion Analysis}
The quantization distortion $D_{k}$ of the user $k$ is defined as 
\begin{align}\label{5.1.1}
D_{k}=\mathbb{E}\left[\left\lvert\sum_{m=1}^{M}\mathbf{h}_{m,k}^{\textrm{H}}\mathbf{w}_{m,k}^{(\text{ideal})}\right\rvert^{2}-\left\lvert\sum_{m=1}^{M}\mathbf{h}_{m,k}^{\textrm{H}}\mathbf{w}_{m,k}\right\rvert^{2}\right]
\end{align}
where $\mathbf{w}_{m,k}^{(\text{ideal})}$ is the precoding vector constructed from the perfect PGI. By plugging \eqref{2.1.3} and \eqref{3.2.1} into \eqref{5.1.1}, we get
\begin{align}
D_{k}&=\mathbb{E}\left[\left\lvert\sum_{m=1}^{M}\mathbf{g}_{m,k}^{\textrm{H}}\mathbf{A}_{m,k}^{\textrm{H}}\mathbf{V}_{\Lambda_{m,k}}\mathbf{g}_{\Lambda_{m,k}}\right\rvert^{2}-\left\lvert\sum_{m=1}^{M}\mathbf{g}_{m,k}^{\textrm{H}}\mathbf{A}_{m,k}^{\textrm{H}}\mathbf{V}_{\Lambda_{m,k}}\hat{\mathbf{g}}_{\Lambda_{m,k}}\right\rvert^{2}\right]\nonumber\\
&\stackrel{(a)}{=}\mathbb{E}\left[\left\lvert\sum_{m=1}^{M}\mathbf{g}_{\Lambda_{m,k}}^{\textrm{H}}\mathbf{A}_{\Lambda_{m,k}}^{\textrm{H}}\mathbf{V}_{\Lambda_{m,k}}\mathbf{g}_{\Lambda_{m,k}}\right\rvert^{2}-\left\lvert\sum_{m=1}^{M}\mathbf{g}_{\Lambda_{m,k}}^{\textrm{H}}\mathbf{A}_{\Lambda_{m,k}}^{\textrm{H}}\mathbf{V}_{\Lambda_{m,k}}\hat{\mathbf{g}}_{\Lambda_{m,k}}\right\rvert^{2}\right]\label{5.1.2.2}
\end{align}
where $(a)$ is due to the fact that $\mathbf{A}_{m,k}\mathbf{g}_{m,k}=\mathbf{A}_{\Lambda_{m,k}}\mathbf{g}_{\Lambda_{m,k}}+\mathbf{A}_{\Lambda_{m,k}^{\mathsf{C}}}\mathbf{g}_{\Lambda_{m,k}^{\mathsf{C}}}$ and $\hat{\mathbf{g}}_{\Lambda_{m,k}}$ is independent with $\mathbf{g}_{\Lambda_{m,k}^{\mathsf{C}}}$. Based on \eqref{5.1.2.2}, the normalized quantization distortion $\bar{D}_{k}$ is given by
\begin{align}\label{5.1.3}
\bar{D}_{k}=\frac{\mathbb{E}\left[\left\lvert\sum_{m=1}^{M}\mathbf{g}_{\Lambda_{m,k}}^{\textrm{H}}\mathbf{A}_{\Lambda_{m,k}}^{\textrm{H}}\mathbf{V}_{\Lambda_{m,k}}\mathbf{g}_{\Lambda_{m,k}}\right\rvert^{2}-\left\lvert\sum_{m=1}^{M}\mathbf{g}_{\Lambda_{m,k}}^{\textrm{H}}\mathbf{A}_{\Lambda_{m,k}}^{\textrm{H}}\mathbf{V}_{\Lambda_{m,k}}\hat{\mathbf{g}}_{\Lambda_{m,k}}\right\rvert^{2}\right]}{\mathbb{E}\left[\left\lvert\sum_{m=1}^{M}\mathbf{g}_{\Lambda_{m,k}}^{\textrm{H}}\mathbf{A}_{\Lambda_{m,k}}^{\textrm{H}}\mathbf{V}_{\Lambda_{m,k}}\mathbf{g}_{\Lambda_{m,k}}\right\rvert^{2}\right]}
\end{align}
In the following proposition, we provide an upper bound of $\bar{D}_{k}$. 
\begin{prop}
The normalized quantization distortion $\bar{D}_{k}$ of the user $k$ is upper bounded
\begin{align}\label{5.1.4}
\bar{D}_{k} \leq \frac{L-\delta_{k}}{(L-1)(1+\delta_{k})}2^{-\frac{B}{L-1}},
\end{align}
where $\delta_{k}=\frac{\sum_{m=1}^{M}\lVert\mathbf{A}_{\Lambda_{m,k}}^{\textup{\textrm{H}}}\mathbf{V}_{\Lambda_{m,k}}\rVert_{\textup{F}}^{2}}{\lvert\sum_{m=1}^{M}\textup{tr}(\mathbf{A}_{\Lambda_{m,k}}^{\textup{\textrm{H}}}\mathbf{V}_{\Lambda_{m,k}})\rvert^{2}}$. Furthermore, $\bar{D}_{k}$ is generally upper bounded as $\bar{D}_{k}\leq 2^{-\frac{B}{L-1}}$.
\end{prop}
\begin{proof}
In order to simplify the expression, we use the notation $\mathbf{A}_{\Lambda_{k}}=\text{diag}\left(\mathbf{A}_{\Lambda_{1,k}},\cdots,\mathbf{A}_{\Lambda_{M,k}}\right)$ and $\mathbf{V}_{\Lambda_{k}}=\text{diag}\left(\mathbf{V}_{\Lambda_{1,k}},\cdots,\mathbf{V}_{\Lambda_{M,k}}\right)$. Then, we have
\begin{align}
\bar{D}_{k}&=\frac{\mathbb{E}\left[\left\lvert\mathbf{g}_{\Lambda_{k}}^{\textrm{H}}\mathbf{A}_{\Lambda_{k}}^{\textrm{H}}\mathbf{V}_{\Lambda_{k}}\mathbf{g}_{\Lambda_{k}}\right\rvert^{2}-\left\lvert\mathbf{g}_{\Lambda_{k}}^{\textrm{H}}\mathbf{A}_{\Lambda_{k}}^{\textrm{H}}\mathbf{V}_{\Lambda_{k}}\hat{\mathbf{g}}_{\Lambda_{k}}\right\rvert^{2}\right]}{\mathbb{E}\left[\left\lvert\mathbf{g}_{\Lambda_{k}}^{\textrm{H}}\mathbf{A}_{\Lambda_{k}}^{\textrm{H}}\mathbf{V}_{\Lambda_{k}}\mathbf{g}_{\Lambda_{k}}\right\rvert^{2}\right]}\nonumber\\
&\stackrel{(a)}{=}1-\frac{\mathbb{E}\left[\left\lVert\mathbf{g}_{\Lambda_{k}}\right\rVert^{4}\left\lvert\bar{\mathbf{g}}_{\Lambda_{k}}^{\textrm{H}}\mathbf{A}_{\Lambda_{k}}^{\textrm{H}}\mathbf{V}_{\Lambda_{k}}\mathbf{c}_{\hat{i}_{k}}\right\rvert^{2}\right]}{\mathbb{E}\left[\left\lVert\mathbf{g}_{\Lambda_{k}}\right\rVert^{4}\left\lvert\bar{\mathbf{g}}_{\Lambda_{k}}^{\textrm{H}}\mathbf{A}_{\Lambda_{k}}^{\textrm{H}}\mathbf{V}_{\Lambda_{k}}\bar{\mathbf{g}}_{\Lambda_{k}}\right\rvert^{2}\right]}\nonumber\\
&=1-\frac{\mathbb{E}\left[\left\lvert\bar{\mathbf{g}}_{\Lambda_{k}}^{\textrm{H}}\mathbf{A}_{\Lambda_{k}}^{\textrm{H}}\mathbf{V}_{\Lambda_{k}}\mathbf{c}_{\hat{i}_{k}}\right\rvert^{2}\right]}{\mathbb{E}\left[\left\lvert\bar{\mathbf{g}}_{\Lambda_{k}}^{\textrm{H}}\mathbf{A}_{\Lambda_{k}}^{\textrm{H}}\mathbf{V}_{\Lambda_{k}}\bar{\mathbf{g}}_{\Lambda_{k}}\right\rvert^{2}\right]}\label{5.1.5.3}
\end{align}
where $(a)$ is due to the independence of the vector norm $\left\lVert\mathbf{g}_{\Lambda_{k}}\right\rVert$ and the vector direction $\bar{\mathbf{g}}_{\Lambda_{k}}$.

Now, we compute the closed-form expression of the nominator $\mathbb{E}\left[\left\lvert\bar{\mathbf{g}}_{\Lambda_{k}}^{\textrm{H}}\mathbf{A}_{\Lambda_{k}}^{\textrm{H}}\mathbf{V}_{\Lambda_{k}}\mathbf{c}_{\hat{i}_{k}}\right\rvert^{2}\right]$ and the denominator $\mathbb{E}\left[\left\lvert\bar{\mathbf{g}}_{\Lambda_{k}}^{\textrm{H}}\mathbf{A}_{\Lambda_{k}}^{\textrm{H}}\mathbf{V}_{\Lambda_{k}}\bar{\mathbf{g}}_{\Lambda_{k}}\right\rvert^{2}\right]$ in \eqref{5.1.5.3}. When the $B$-bit RVQ codebook $\mathcal{C}_{k}=\lbrace \mathbf{c}_{1},\cdots,\mathbf{c}_{2^{B}}\rbrace$ is used, the correlation $\left\lvert\bar{\mathbf{g}}_{\Lambda_{k}}^{\textrm{H}}\mathbf{c}_{\hat{i}_{k}}\right\rvert^{2}$ between the dominating PGI direction $\bar{\mathbf{g}}_{\Lambda_{k}}$ and the chosen codeword $\mathbf{c}_{\hat{i}_{k}}$ is $\beta$-distributed random variable with parameters $1$ and $L-1$~\cite{jindal2006mimo}. That is
\begin{align}\label{5.1.6}
1-\mathbb{E}\left[\left\lvert\bar{\mathbf{g}}_{\Lambda_{k}}^{\textrm{H}}\mathbf{c}_{\hat{i}_{k}}\right\rvert^{2}\right]=2^{B}\beta\left(2^{B},\frac{L}{L-1}\right)\leq 2^{-\frac{B}{L-1}}
\end{align}
Unfortunately, we cannot directly use this result since $\mathbf{A}_{\Lambda_{k}}^{\textrm{H}}\mathbf{V}_{\Lambda_{k}}$ is inserted in the middle of $\mathbb{E}\left[\left\lvert\bar{\mathbf{g}}_{\Lambda_{k}}^{\textrm{H}}\mathbf{A}_{\Lambda_{k}}^{\textrm{H}}\mathbf{V}_{\Lambda_{k}}\mathbf{c}_{\hat{i}_{k}}\right\rvert^{2}\right]$. To handle this, we exploit the property that the dominating PGI direction $\bar{\mathbf{g}}_{\Lambda_{k}}$ can be written as a sum of two vectors: one in the direction of the chosen codeword $\mathbf{c}_{\hat{i}_{k}}$ and the other isotropically distributed in the null space of $\mathbf{c}_{\hat{i}_{k}}$~\cite{jindal2006mimo}:
\begin{align}\label{5.1.7}
\bar{\mathbf{g}}_{\Lambda_{k}}=\sqrt{Z}\mathbf{c}_{\hat{i}_{k}}+\sqrt{1-Z}\mathbf{s}
\end{align}
where $\mathbf{s}$ is a unit norm vector isotropically distributed in the null space of $\mathbf{c}_{\hat{i}_{k}}$ and $Z$ is $\beta$-distributed according to $\left\lvert\bar{\mathbf{g}}_{\Lambda_{k}}^{\textrm{H}}\mathbf{c}_{\hat{i}_{k}}\right\rvert^{2}$. Also, $\mathbf{s}$ and $Z$ are independent. Using \eqref{5.1.7}, we obtain
\begin{align}
\mathbb{E}\left[\left\lvert\bar{\mathbf{g}}_{\Lambda_{k}}^{\textrm{H}}\mathbf{A}_{\Lambda_{k}}^{\textrm{H}}\mathbf{V}_{\Lambda_{k}}\mathbf{c}_{\hat{i}_{k}}\right\rvert^{2}\right]&=\mathbb{E}\left[\mathbf{c}_{\hat{i}_{k}}^{\textrm{H}}\mathbf{V}_{\Lambda_{k}}^{\textrm{H}}\mathbf{A}_{\Lambda_{k}}\left(Z\mathbf{c}_{\hat{i}_{k}}\mathbf{c}_{\hat{i}_{k}}^{\textrm{H}}+\left(1-Z\right)\mathbf{s}\mathbf{s}^{\textrm{H}}\right)\mathbf{A}_{\Lambda_{k}}^{\textrm{H}}\mathbf{V}_{\Lambda_{k}}\mathbf{c}_{\hat{i}_{k}}\right]\label{5.1.8.1}\\
&=\gamma\mathbb{E}\left[\left\lvert\mathbf{c}_{\hat{i}_{k}}^{\textrm{H}}\mathbf{A}_{\Lambda_{k}}^{\textrm{H}}\mathbf{V}_{\Lambda_{k}}\mathbf{c}_{\hat{i}_{k}}\right\rvert^{2}\right]+\left(1-\gamma\right)\mathbb{E}\left[\left\lvert\mathbf{s}^{\textrm{H}}\mathbf{A}_{\Lambda_{k}}^{\textrm{H}}\mathbf{V}_{\Lambda_{k}}\mathbf{c}_{\hat{i}_{k}}\right\rvert^{2}\right]\label{5.1.8.2}
\end{align}
where $\gamma=\mathbb{E}\left[Z\right]=\mathbb{E}\left[\left\lvert\bar{\mathbf{g}}_{\Lambda_{k}}^{\textrm{H}}\mathbf{c}_{\hat{i}_{k}}\right\rvert^{2}\right]=1-2^{B}\beta\left(2^{B},\frac{L}{L-1}\right)$ in \eqref{5.1.6}. Using Lemma 2 (see Appendix A), we obtain the closed-form expression of the first term in \eqref{5.1.8.2} as
\begin{align}\label{5.1.9}
\mathbb{E}\left[\left\lvert\mathbf{c}_{\hat{i}_{k}}^{\textrm{H}}\mathbf{A}_{\Lambda_{k}}^{\textrm{H}}\mathbf{V}_{\Lambda_{k}}\mathbf{c}_{\hat{i}_{k}}\right\rvert^{2}\right]=\frac{1}{L(L+1)}\left(\left\lvert\text{tr}\left(\mathbf{A}_{\Lambda_{k}}^{\textrm{H}}\mathbf{V}_{\Lambda_{k}}\right)\right\rvert^{2}+\left\lVert\mathbf{A}_{\Lambda_{k}}^{\textrm{H}}\mathbf{V}_{\Lambda_{k}}\right\rVert_{\text{F}}^{2}\right)
\end{align}
Whereas, since $\mathbf{s}$ is in the null space of $\mathbf{c}_{\hat{i}_{k}}$, $\mathbf{s}$ and $\mathbf{c}_{\hat{i}_{k}}$ are correlated, and thus it is not easy to obtain the closed-form expression of the second term in \eqref{5.1.8.2}. As a remedy, we use the law of total expectation, that is
\begin{align}
\mathbb{E}_{\mathbf{s},\mathbf{c}_{\hat{i}_{k}}}\left[\left\lvert\mathbf{s}^{\textrm{H}}\mathbf{A}_{\Lambda_{k}}^{\textrm{H}}\mathbf{V}_{\Lambda_{k}}\mathbf{c}_{\hat{i}_{k}}\right\rvert^{2}\right]&=\mathbb{E}_{\mathbf{c}_{\hat{i}_{k}}}\left[\mathbb{E}_{\mathbf{s}}\left[\left\lvert\mathbf{s}^{\textrm{H}}\mathbf{A}_{\Lambda_{k}}^{\textrm{H}}\mathbf{V}_{\Lambda_{k}}\mathbf{c}_{\hat{i}_{k}}\right\rvert^{2}\mid\mathbf{c}_{\hat{i}_{k}}\right]\right]\label{5.1.10.1}\\
&=\mathbb{E}_{\mathbf{c}_{\hat{i}_{k}}}\left[\mathbf{c}_{\hat{i}_{k}}^{\textrm{H}}\mathbf{V}_{\Lambda_{k}}^{\textrm{H}}\mathbf{A}_{\Lambda_{k}}\mathbb{E}_{\mathbf{s}}\left[\mathbf{s}\mathbf{s}^{\textrm{H}}\mid \mathbf{c}_{\hat{i}_{k}}\right]\mathbf{A}_{\Lambda_{k}}\mathbf{V}_{\Lambda_{k}}^{\textrm{H}}\mathbf{c}_{\hat{i}_{k}}\right]\label{5.1.10.2}
\end{align}
In the following lemma, we provide the conditional covariance of $\mathbf{s}$ for a given $\mathbf{c}_{\hat{i}_{k}}$.
\begin{lemma}
The conditional covariance of $\mathbf{s}$ for a given $\mathbf{c}_{\hat{i}_{k}}$ is
\begin{align}\label{5.1.11}
\mathbb{E}_{\mathbf{s}}\left[\mathbf{s}\mathbf{s}^{\textrm{H}}\mid \mathbf{c}_{\hat{i}_{k}}\right]=\frac{1}{L-1}\left(\mathbf{I}_{L}-\mathbf{c}_{\hat{i}_{k}}\mathbf{c}_{\hat{i}_{k}}^{\textrm{H}}\right)
\end{align} 
\end{lemma}
\begin{proof}
See Appendix B.
\end{proof}
\noindent By plugging \eqref{5.1.11} into the second term of \eqref{5.1.8.2}, we obtain
\begin{align}
\mathbb{E}\left[\left\lvert\mathbf{s}^{\textrm{H}}\mathbf{A}_{\Lambda_{k}}^{\textrm{H}}\mathbf{V}_{\Lambda_{k}}\mathbf{c}_{\hat{i}_{k}}\right\rvert^{2}\right]&=\frac{1}{L-1}\mathbb{E}_{\mathbf{c}_{\hat{i}_{k}}}\left[\mathbf{c}_{\hat{i}_{k}}^{\textrm{H}}\mathbf{V}_{\Lambda_{k}}^{\textrm{H}}\mathbf{A}_{\Lambda_{k}}\left(\mathbf{I}_{L}-\mathbf{c}_{\hat{i}_{k}}\mathbf{c}_{\hat{i}_{k}}^{\textrm{H}}\right)\mathbf{A}_{\Lambda_{k}}\mathbf{V}_{\Lambda_{k}}^{\textrm{H}}\mathbf{c}_{\hat{i}_{k}}\right]\nonumber\\
&=\frac{1}{L-1}\left(\mathbb{E}_{\mathbf{c}_{\hat{i}_{k}}}\left[\left\lvert\mathbf{A}_{\Lambda_{k}}^{\textrm{H}}\mathbf{V}_{\Lambda_{k}}\mathbf{c}_{\hat{i}_{k}}\right\rvert^{2}\right]-\mathbb{E}_{\mathbf{c}_{\hat{i}_{k}}}\left[\left\lvert\mathbf{c}_{\hat{i}_{k}}^{\textrm{H}}\mathbf{A}_{\Lambda_{k}}^{\textrm{H}}\mathbf{V}_{\Lambda_{k}}\mathbf{c}_{\hat{i}_{k}}\right\rvert^{2}\right]\right)\nonumber\\
&=\frac{1}{L-1}\left(\frac{1}{L}\left\lVert\mathbf{A}_{\Lambda_{k}}^{\textrm{H}}\mathbf{V}_{\Lambda_{k}}\right\rVert_{\text{F}}^{2}-\frac{1}{L(L+1)}\left(\left\lvert\text{tr}\left(\mathbf{A}_{\Lambda_{k}}^{\textrm{H}}\mathbf{V}_{\Lambda_{k}}\right)\right\rvert^{2}+\left\lVert\mathbf{A}_{\Lambda_{k}}^{\textrm{H}}\mathbf{V}_{\Lambda_{k}}\right\rVert_{\text{F}}^{2}\right)\right)\nonumber\\
&=\frac{1}{L^{2}-1}\left(\left\lVert\mathbf{A}_{\Lambda_{k}}^{\textrm{H}}\mathbf{V}_{\Lambda_{k}}\right\rVert_{\text{F}}^{2}-\frac{1}{L}\left\lvert\text{tr}\left(\mathbf{A}_{\Lambda_{k}}^{\textrm{H}}\mathbf{V}_{\Lambda_{k}}\right)\right\rvert^{2}\right)\label{5.1.12.3}
\end{align}
Finally, by plugging \eqref{5.1.9} and \eqref{5.1.12.3} into \eqref{5.1.8.2}, we get
\begin{align}\label{5.1.13}
\mathbb{E}\left[\left\lvert\bar{\mathbf{g}}_{\Lambda_{k}}^{\textrm{H}}\mathbf{A}_{\Lambda_{k}}^{\textrm{H}}\mathbf{V}_{\Lambda_{k}}\mathbf{c}_{\hat{i}_{k}}\right\rvert^{2}\right]=&\gamma\frac{1}{L(L+1)}\left(\left\lvert\text{tr}\left(\mathbf{A}_{\Lambda_{k}}^{\textrm{H}}\mathbf{V}_{\Lambda_{k}}\right)\right\rvert^{2}+\left\lVert\mathbf{A}_{\Lambda_{k}}^{\textrm{H}}\mathbf{V}_{\Lambda_{k}}\right\rVert_{\text{F}}^{2}\right)\nonumber\\
&+\left(1-\gamma\right)\frac{1}{L^{2}-1}\left(\left\lVert\mathbf{A}_{\Lambda_{k}}^{\textrm{H}}\mathbf{V}_{\Lambda_{k}}\right\rVert_{\textup{F}}^{2}-\frac{1}{L}\left\lvert\textup{tr}\left(\mathbf{A}_{\Lambda_{k}}^{\textrm{H}}\mathbf{V}_{\Lambda_{k}}\right)\right\rvert^{2}\right)
\end{align}

Next, we consider $\mathbb{E}\left[\left\lvert\bar{\mathbf{g}}_{\Lambda_{k}}^{\textrm{H}}\mathbf{A}_{\Lambda_{k}}^{\textrm{H}}\mathbf{V}_{\Lambda_{k}}\bar{\mathbf{g}}_{\Lambda_{k}}\right\rvert^{2}\right]$ in \eqref{5.1.5.3}. Since both $\bar{\mathbf{g}}_{\Lambda_{k}}$ and $\mathbf{c}_{\hat{i}_{k}}$ distributed uniformly on the surface of a $L$-dimensional unit sphere, the closed-form expression of $\mathbb{E}\left[\left\lvert\bar{\mathbf{g}}_{\Lambda_{k}}^{\textrm{H}}\mathbf{A}_{\Lambda_{k}}^{\textrm{H}}\mathbf{V}_{\Lambda_{k}}\bar{\mathbf{g}}_{\Lambda_{k}}\right\rvert^{2}\right]$ can be obtained in the same way to \eqref{5.1.9}. Finally, the closed-form expression and the upper bound of $\bar{D}_{k}$ is 
\begin{align}
\bar{D}_{k}&=1-\frac{\frac{\gamma}{L(L+1)}\left(\left\lvert\text{tr}\left(\mathbf{A}_{\Lambda_{k}}^{\textrm{H}}\mathbf{V}_{\Lambda_{k}}\right)\right\rvert^{2}+\left\lVert\mathbf{A}_{\Lambda_{k}}^{\textrm{H}}\mathbf{V}_{\Lambda_{k}}\right\rVert_{\text{F}}^{2}\right)+\frac{1-\gamma}{L^{2}-1}\left(\left\lVert\mathbf{A}_{\Lambda_{k}}^{\textrm{H}}\mathbf{V}_{\Lambda_{k}}\right\rVert_{\textup{F}}^{2}-\frac{1}{L}\left\lvert\textup{tr}\left(\mathbf{A}_{\Lambda_{k}}^{\textrm{H}}\mathbf{V}_{\Lambda_{k}}\right)\right\rvert^{2}\right)}{\frac{1}{L(L+1)}\left(\left\lvert\text{tr}\left(\mathbf{A}_{\Lambda_{k}}^{\textrm{H}}\mathbf{V}_{\Lambda_{k}}\right)\right\rvert^{2}+\left\lVert\mathbf{A}_{\Lambda_{k}}^{\textrm{H}}\mathbf{V}_{\Lambda_{k}}\right\rVert_{\text{F}}^{2}\right)}\nonumber\\
&=\left(1-\gamma\right)\frac{1}{L-1}\frac{L\left\lvert\text{tr}\left(\mathbf{A}_{\Lambda_{k}}^{\textrm{H}}\mathbf{V}_{\Lambda_{k}}\right)\right\rvert^{2}-\left\lVert\mathbf{A}_{\Lambda_{k}}^{\textrm{H}}\mathbf{V}_{\Lambda_{k}}\right\rVert_{\text{F}}^{2}}{\left\lvert\text{tr}\left(\mathbf{A}_{\Lambda_{k}}^{\textrm{H}}\mathbf{V}_{\Lambda_{k}}\right)\right\rvert^{2}+\left\lVert\mathbf{A}_{\Lambda_{k}}^{\textrm{H}}\mathbf{V}_{\Lambda_{k}}\right\rVert_{\text{F}}^{2}}\nonumber\\
&=\left(1-\gamma\right)\frac{L-\delta_{k}}{(L-1)(1+\delta_{k})}\nonumber\\
&\stackrel{(a)}{\leq} 2^{-\frac{B}{L-1}}\frac{L-\delta_{k}}{(L-1)(1+\delta_{k})}\label{5.1.14.3}
\end{align}
where $\delta_{k}=\frac{\sum_{m=1}^{M}\lVert\mathbf{A}_{\Lambda_{m,k}}^{\textrm{H}}\mathbf{V}_{\Lambda_{m,k}}\rVert_{\textup{F}}^{2}}{\lvert\sum_{m=1}^{M}\textup{tr}(\mathbf{A}_{\Lambda_{m,k}}^{\textrm{H}}\mathbf{V}_{\Lambda_{m,k}})\rvert^{2}}$ and $(a)$ is due to \eqref{5.1.6}. By using that $\frac{1}{L}\leq \frac{\left\lVert\mathbf{C}\right\rVert_{\text{F}}^{2}}{\left\lvert\text{tr}\left(\mathbf{C}\right)\right\rvert^{2}}$, we can obtain a simple upper bound of $\bar{D}_{k}$ as
\begin{align}\label{5.1.15}
\bar{D}_{k}\leq \frac{L-\frac{1}{L}}{(L-1)(1+\frac{1}{L})}2^{-\frac{B}{L-1}}=2^{-\frac{B}{L-1}}
\end{align}
\end{proof}
Since $\frac{1}{L}\leq\delta_{k}$, we can observe that $\bar{D}_{k}$ is smaller than the normalized quantization distortion of the conventional $L$-dimensional vector quantization, that is $1-\gamma$ in \eqref{5.1.6}. It is worth mentioning that $\bar{D}_{k}$ is a function of the number of dominating paths $L$, not the number of transmit antennas $N$. In Fig. 6, we plot the normalized quantization distortion $\bar{D}_{k}$ as a function of the number of dominating paths $L$. We plot the numerical evaluation of $\bar{D}_{k}$, the upper bound in \eqref{5.1.14.3}, the simplified upper bound in \eqref{5.1.15}, and the conventional $L$-dimensional vector quantization using RVQ codebook in \eqref{5.1.6}. One can observe that the numerical evaluation is close to the derived upper bound. One can also observe that the quantization distortion of the proposed scheme is much smaller than that of the conventional vector quantization.

\begin{figure}[h]
\centering
\includegraphics[scale=0.68]{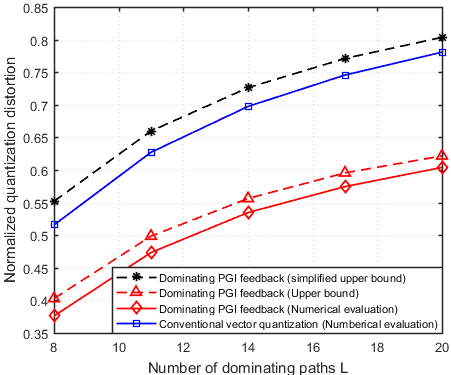}
\caption{Normalized quantization distortion as a function of the number of dominating paths $L$ ($M=5$, $N=8$, $P=4$, $B=6$, $\text{SNR}=15\,\text{dB}$)}
\vspace{-1em}
\end{figure}

\subsection{Rate Gap Analysis of the Dominating PGI Feedback}
In this subsection, we analyze the per user rate gap of the dominating PGI feedback scheme between the ideal feedback system and the finite rate feedback system. 

\begin{theorem}
The per user rate gap $\Delta R_{k}$ between the ideal system using the perfect PGI and the realistic system using the finite rate feedback of the user $k$ is upper bounded as
\begin{align}\label{5.2.1}
\Delta R_{k}\leq \log_{2}\left(1+\frac{\textup{SNR}}{1+\textup{SNR}}\frac{L-\delta_{k}}{(L-1)(1+\delta_{k})-2^{-\frac{B}{L-1}}\left(L-\delta_{k}\right)}2^{-\frac{B}{L-1}}\right)
\end{align}
where $\textup{SNR}$ is the signal-to-noise-ratio.
\end{theorem}
\begin{proof}
The achievable rate $R_{k}$ of the user $k$ in the realistic system with finite rate feedback is
\begin{align}
R_{k}&=\log_{2}\left(1+\frac{\mathbb{E}\left[\left\lvert\sum_{m=1}^{M}\mathbf{h}_{m,k}^{\textrm{H}}\mathbf{w}_{m,k}\right\rvert^{2}\right]}{\sum_{j\neq k}^{K}\mathbb{E}\left[\left\lvert\sum_{m=1}^{M}\mathbf{h}_{m,k}^{\textrm{H}}\mathbf{w}_{m,j}\right\rvert^{2}\right]+\sigma_{n}^{2}}\right)\nonumber\\
&=\log_{2}\left(1+\frac{\overbracket{\mathbb{E}\left[\left\lvert\sum_{m=1}^{M}\mathbf{g}_{\Lambda_{m,k}}^{\textrm{H}}\mathbf{A}_{\Lambda_{m,k}}^{\textrm{H}}\mathbf{V}_{\Lambda_{m,k}}\hat{\mathbf{g}}_{\Lambda_{m,k}}\right\rvert^{2}\right]}^{\text{DS}_{k}}+\overbracket{\mathbb{E}\left[\left\lvert\sum_{m=1}^{M}\mathbf{g}_{\Lambda_{m,k}^{\mathsf{C}}}^{\textrm{H}}\mathbf{A}_{\Lambda_{m,k}^{\mathsf{C}}}^{\textrm{H}}\mathbf{V}_{\Lambda_{m,k}}\hat{\mathbf{g}}_{\Lambda_{m,k}}\right\rvert^{2}\right]}^{\text{US}_{k}}}{\underbracket{\sum_{j\neq k}^{K}\mathbb{E}\left[\left\lvert\sum_{m=1}^{M}\mathbf{g}_{m,k}^{\textrm{H}}\mathbf{A}_{m,k}^{\textrm{H}}\mathbf{V}_{\Lambda_{m,j}}\hat{\mathbf{g}}_{\Lambda_{m,j}}\right\rvert^{2}\right]}_{\text{IS}_{k}}+\sigma_{n}^{2}}\right)\nonumber
\end{align}
Note that $R_{k}$ consists of the desired signal part $\text{DS}_{k}$, the unselected signal part $\text{US}_{k}$, and interference signal part $\text{IS}_{k}$, respectively. Since $\mathbf{g}_{\Lambda_{m,k}}$ is independent with $\mathbf{g}_{\Lambda_{m,k}^{\mathsf{C}}}$ and $\mathbf{g}_{m,j}$ ($j\neq k$), $\hat{\mathbf{g}}_{\Lambda_{m,k}}$ is also independent with with $\mathbf{g}_{\Lambda_{m,k}^{\mathsf{C}}}$ and $\mathbf{g}_{m,j}$ ($j\neq k$) so that the quantization of $\mathbf{g}_{\Lambda_{m,k}}$ only affects $\text{DS}_{k}$. This means that $\text{US}_{k}$ and $\text{IS}_{k}$ remain unchanged regardless of the quantization. Based on this observation, the achievable user rates for the realistic system $R_{k}$ and the ideal system $R_{k}^{(\text{ideal})}$ are given by
\begin{align}\label{5.2.4}
R_{k}&=\log_{2}\left(1+\frac{\text{DS}_{k}+\text{US}_{k}}{\text{IS}_{k}+\sigma_{n}^{2}}\right)\\
R_{k}^{(\text{ideal})}&=\log_{2}\left(1+\frac{\text{DS}_{k}^{(\text{ideal})}+\text{US}_{k}}{\text{IS}_{k}+\sigma_{n}^{2}}\right)
\end{align}
where $\text{DS}_{k}^{(\text{ideal})}$ is the desired signal part constructed from the perfect PGI. Thus, the rate gap $\Delta R_{k}=R_{k}^{(\text{ideal})}-R_{k}$ is
\begin{align}\label{5.2.5}
\Delta R_{k}&=\log_{2}\left(1+\frac{\text{DS}_{k}^{(\text{ideal})}+\text{US}_{k}}{\text{IS}_{k}+\sigma_{n}^{2}}\right)-\log_{2}\left(1+\frac{\text{DS}_{k}+\text{US}_{k}}{\text{IS}_{k}+\sigma_{n}^{2}}\right)\\
&=\log_{2}\left(1+\frac{\text{DS}_{k}^{(\text{ideal})}-\text{DS}_{k}}{\text{DS}_{k}+\text{US}_{k}+\text{IS}_{k}+\sigma_{n}^{2}}\right)
\end{align}
From $\text{DS}_{k}^{(\text{ideal})}-\text{DS}_{k}=\text{DS}_{k}^{(\text{ideal})}\bar{D}_{k}$ we get $\text{DS}_{k}^{(\text{ideal})}=\frac{\text{DS}_{k}}{1-\bar{D}_{k}}$. Using this, together with Proposition 1, we have
\begin{align}
\Delta R_{k}&=\log_{2}\left(1+\frac{\bar{D}_{k}}{1-\bar{D}_{k}}\frac{\text{DS}_{k}}{\text{DS}_{k}+\text{US}_{k}+\text{IS}_{k}+\sigma_{n}^{2}}\right)\nonumber\\
&\stackrel{(a)}{=}\log_{2}\left(1+\frac{\bar{D}_{k}}{1-\bar{D}_{k}}\frac{\text{DS}_{k}}{\left(1+\frac{1}{\text{SNR}}\right)\left(\text{DS}_{k}+\text{US}_{k}+\text{IS}_{k}\right)}\right)\nonumber\\
&\leq \log_{2}\left(1+\frac{\bar{D}_{k}}{1-\bar{D}_{k}}\frac{\text{SNR}}{1+\text{SNR}}\right)\nonumber\\
&\stackrel{(b)}{\leq} \log_{2}\left(1+\frac{\text{SNR}}{1+\text{SNR}}\frac{2^{-\frac{B}{L-1}}\left(L-\delta_{k}\right)}{(L-1)(1+\delta_{k})-2^{-\frac{B}{L-1}}\left(L-\delta_{k}\right)}\right)\nonumber
\end{align}
where $(a)$ is because $\text{SNR}=\frac{\text{DS}_{k}+\text{US}_{k}+\text{IS}_{k}}{\sigma_{n}^{2}}$ and $(b)$ is from Proposition 1.
\end{proof}

Finally, we can obtain the number of feedback bits required to maintain a constant rate gap with the ideal system.
\begin{prop}
To maintain a constant rate gap with the ideal system with perfect PGI within $\log_{2}\left(\beta\right)\,\textup{bps}/\textup{Hz}$ per user, it is sufficient to scale the number of bits per user according to
\begin{align}\label{5.2.7}
B=(L-1)\left(\log_{2}\left(\frac{\textup{SNR}}{\left(\textup{SNR}+1\right)\left(\beta-1\right)}-1\right)+\log_{2}\left(\frac{L-\delta_{k}}{\left(L-1\right)\left(1+\delta_{k}\right)}\right)\right)
\end{align} 
\end{prop}
\begin{proof}
To maintain a rate gap of $\Delta R_{k}\leq \log_{2}\left(\beta\right)$, the number of feedback bits $B$ should satisfy
\begin{align}\label{5.2.8}
\Delta R_{k}\leq \log_{2}\left(1+\frac{\text{SNR}}{1+\text{SNR}}\frac{L-\delta_{k}}{(L-1)(1+\delta_{k})-2^{-\frac{B}{L-1}}\left(L-\delta_{k}\right)}2^{-\frac{B}{L-1}}\right)=\log_{2}\left(\beta\right).
\end{align}
After simple manipulations, we get the desired result.
\end{proof}
In Fig. 7, we plot the per user rate as a function of SNR. We observe that the analytic upper bound obtained from the Theorem 2 is close to the upper bound obtained from the numerical evaluation. This means that by using a proper scaling of feedback bits in Proposition 2, the rate loss can be controlled effectively. 
\begin{figure}[h]
\centering
\includegraphics[scale=0.7]{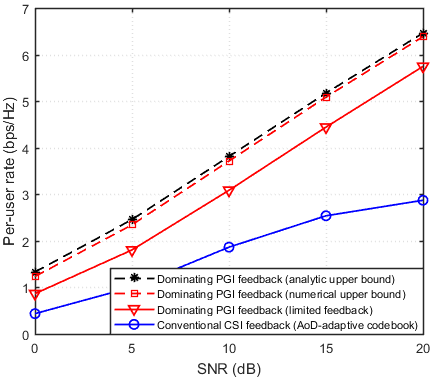}
\caption{Per user rate as a function of SNR ($M=5$, $N=8$, $P=4$, $L=8$, $B=6$)}
\vspace{-1em}
\end{figure}
\subsection{Dominating Path Number Selection}
In the subsection, we discuss how to choose the dominating path number. In a nutshell, we compute the lower bound of the sum rate $\sum_{k=1}^{K}R_{k}\left(l\right)$ for each $l$ ($l=1,\cdots, MP$) and then choose the value $L$ maximizing the sum rate. That is
\begin{align}
L=\text{arg }\underset{l=1,\cdots,MP}{\text{max}}\,\sum_{k=1}^{K}R_{k}\left(l\right).
\end{align}
Note that $R_{k}\left(l\right)$ is obtained from the dominating path selection algorithm. In each iteration of this algorithm (see Section III.C), we obtain the dominating path indices $\lbrace \Lambda_{m,k}\rbrace$ and the precoding matrices $\lbrace \mathbf{V}_{\Lambda_{m,k}}\rbrace$ and then compute the lower bound of the achievable rate using $\lbrace \Lambda_{m,k}\rbrace$ and $\lbrace \mathbf{V}_{\Lambda_{m,k}}\rbrace$\footnote{To be specific, the lower bound of the rate is $R_{k}(l)=R_{k}^{(\text{ideal})}(l)-\Delta R_{k}(l)$ where $R_{k}^{(\text{ideal})}(l)$ is the rate of ideal system with perfect PGI (see Theorem 1) and $\Delta R_{k}(l)$ is the upper bound of the rate gap over the ideal system (see Theorem 2).}.

Since the dominating path selection depends on AoD information, it is in general very difficult to express the sum rate as a function of $L$. However, in a single cell massive MIMO systems where a macro cell serves users in a cell, we can express the lower bound of sum rate as a function of $L$.
\begin{theorem}
The per user rate $R_{k}$ of the user $k$ in the single cell massive MIMO systems using the dominating path number $L$ is lower bounded as 
\begin{align}\label{5.3.1}
R_{k}\geq \log_{2}\!\left(\!1+\frac{L+1}{\sigma_{n}^{2}}\!\right)\!-\log_{2}\!\left(\!1+\frac{\textup{SNR}}{1+\textup{SNR}}\frac{L-\delta_{k}}{(L-1)(1+\delta_{k})-2^{-\frac{B}{L-1}}\left(L-\delta_{k}\right)}2^{-\frac{B}{L-1}}\!\right).
\end{align}
\end{theorem}
\begin{proof}
See Appendix C.
\end{proof}
\noindent By using Theorem 3, we can easily find out $L$ maximizing the lower bound of sum rate.
\\

\section{Simulation Results}
In this section, we investigate the sum rate performance of the proposed dominating PGI feedback scheme. For comparison, we use the conventional CSI feedback schemes with the AoD-adaptive subspace codebook~\cite{shen2018channel} and the RVQ codebook~\cite{jindal2006mimo}. In our simulations, we consider the FDD-based cell-free systems where $M=5$ (except for Fig. 12) BSs equipped with $N=8$ transmit antennas cooperatively serve $K=5$ users equipped with a single antenna. We set the maximum transmit power of BS to $10\,\text{W}$ and the total transmit power of cooperating BS group to $25\,\text{W}$. Also, we distribute the BSs and users randomly in a square area (size of a square is $1\times 1\,\text{km}^{2}$). We use the downlink narrowband multi-path channel model whose carrier frequency is $f_{c}=2\,\text{GHz}$ and set the number of propagation paths to $P=4$ (except for Fig. 11). The angular spread of AoD is set to $10^{\circ}$. In the proposed dominating PGI feedback scheme, we select $L=8$ (except for Fig. 10) dominating paths among all possible $MP=20$ paths. Further, the number of feedback bits per user is $B=6$ (except for Fig. 9). In order to avoid special scenarios where the proposed technique is favorable (or unfavorable), we used $1000$ randomly generated cell-free system realizations.

\begin{figure}[t]
\begin{minipage}[t]{0.487\linewidth}
    \includegraphics[width=\linewidth]{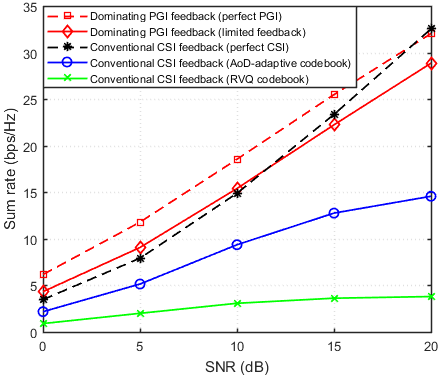}
	\caption{Sum rate as a function of SNR ($M=5$, $K=5$, $N=8$, $P=4$, $L=8$, $B=6$)}
\end{minipage} 
    \hfill%
\begin{minipage}[t]{0.485\linewidth}
    \includegraphics[width=\linewidth]{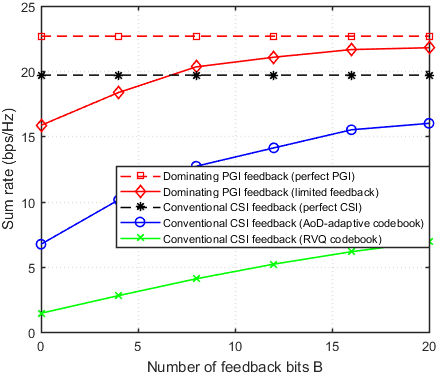}
    \caption{Sum rate as a function of the number of feedback bits $B$ ($M=5$, $K=5$, $N=8$, $P=4$, $L=8$, $\text{SNR}=15\,\text{dB}$)}
\end{minipage} 
\vspace{-1em}
\end{figure}

In Fig. 8, we plot the sum rate performance as a function of SNR. The performance of ideal system with perfect PGI (or CSI) and the realistic system with finite rate feedback are plotted as a dotted line and a real line, respectively. We observe that the proposed dominating PGI feedback scheme outperforms the conventional schemes by a large margin. For example, at $15\,\text{bps/Hz}$ region, the proposed scheme achieves more than $10\,\text{dB}$ gain over the conventional CSI feedback scheme. We also observe that the performance loss of the proposed scheme over the perfect PGI system is within $3\,\text{dB}$ whereas the conventional AoD-adaptive codebook scheme and the RVQ codebook scheme suffer more than $5\,\text{dB}$ and $10\,\text{dB}$ loss. As mentioned, this is because the number of feedback bits in the proposed scheme required to maintain a constant rate gap with the ideal system scales linearly with the number of dominating paths $L$ while such is not the case for the conventional schemes. In fact, with only $B=6$ feedback bits, the proposed scheme performs similar to the conventional feedback scheme with the perfect CSI.

In Fig. 9, we set $\text{SNR}=15\,\text{dB}$ and plot the sum rate as a function of the number of feedback bits $B$. We observe that the proposed dominating PGI feedback scheme achieves a significant feedback overhead reduction over the conventional schemes. For example, in achieving $18\,\text{bps/Hz}$, the proposed dominating PGI feedback scheme requires $B=4$ bits while the AoD-adaptive subspace codebook scheme requires more than $B=20$ bits, resulting in more than $80\%$ reduction in feedback overhead). Further, the proposed scheme requires only $B=8$ bits to maintain $3\,\text{bps/Hz}$ rate gap with the ideal system while the conventional AoD-adaptive codebook scheme requires $B=20$ bits to maintain the same rate gap.

\begin{figure}[t]
\begin{minipage}[t]{0.487\linewidth}
    \includegraphics[width=\linewidth]{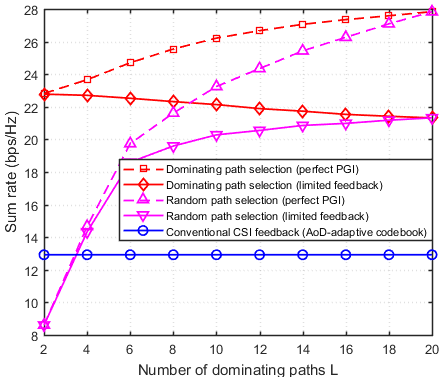}
    \caption{Sum rate as a function of the number of dominating paths $L$ ($M=5$, $K=5$, $N=8$, $P=4$, $B=6$, $\text{SNR}=15\,\text{dB}$)}
\end{minipage} 
    \hfill%
\begin{minipage}[t]{0.485\linewidth}
    \includegraphics[width=\linewidth]{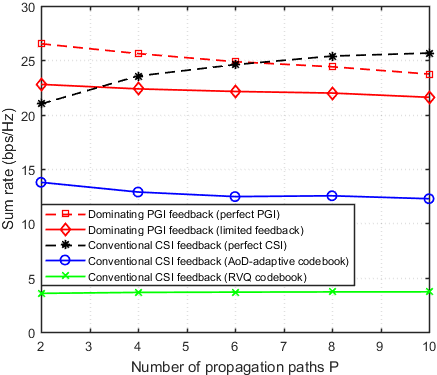}
    \caption{Sum rate as a function of the number of propagation paths $P$ ($M=5$, $K=5$, $N=8$, $L=2P$, $B=6$, $\text{SNR}=15\,\text{dB}$)}
\end{minipage} 
\vspace{-1em}
\end{figure}

In order to show the effectiveness of the dominating path selection, we compare the proposed dominating path selection with the random path selection. By the random path selection, we mean an approach to feed back the PGI of randomly selected paths. The total number of paths is set to $MP=20$. We measure the sum rate as a function of the number of selected paths $L$. Overall, we observe that the dominating path selection provides a considerable sum rate gain over the random path selection approach. When $L=8$, for example, the PGI feedback with dominating path selection achieves $4\,\text{bps/Hz}$ sum rate gain over the PGI feedback with random path selection. We also observe that the performance gain of the proposed scheme increases when the number of dominating paths is small. 

In Fig. 11, we plot the sum rate as a function of the number of propagation paths $P$. In this simulation, we set $\text{SNR}=15\,\text{dB}$ and $L=2P$ so that the number of dominating paths increases linearly with the number of propagation paths. Although the sum rate of the proposed dominating PGI feedback scheme decreases with $P$, the rate loss is not too large even in the rich scattering environment. In fact, when $P$ increases from $2$ to $10$, the rate loss of the proposed scheme is less than $1\,\text{bps/Hz}$.

\begin{figure}[t]
\begin{minipage}[t]{0.485\linewidth}
    \includegraphics[width=\linewidth]{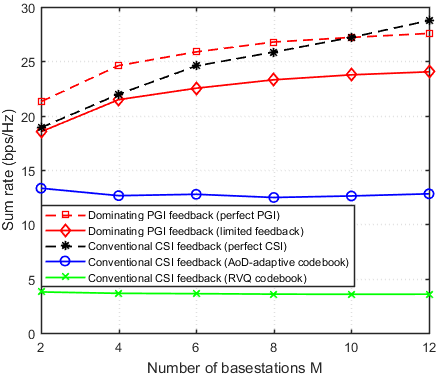}
    \caption{Sum rate as a function of the number of BSs $M$ ($K=5$, $N=8$, $P=4$, $L=8$, $B=6$, $\text{SNR}=15\,\text{dB}$)}
\end{minipage}
    \hfill%
\begin{minipage}[t]{0.49\linewidth}
    \includegraphics[width=\linewidth]{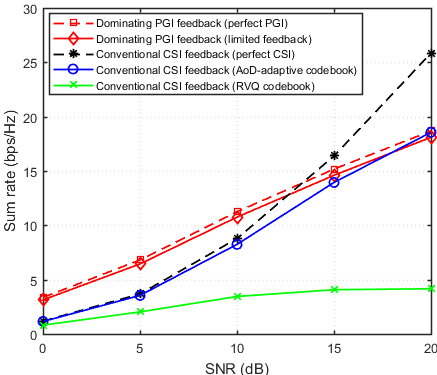}
\caption{Sum rate as a function of SNR ($M=1$, $K=5$, $N=8$, $P=8$, $L=4$, $B=6$)}
\end{minipage} 
\vspace{-1em}
\end{figure}

In Fig. 12, we plot the sum rate as a function of the number of BSs when $\text{SNR}=15\,\text{dB}$. We observe that the sum rate of the proposed dominating PGI feedback scheme increases dramatically with the number of BSs whereas no such effect can be expected from the conventional CSI feedback schemes. In particular, when $M=2$, the rate gap between the dominating PGI scheme and the CSI feedback scheme is $5\,\text{bps/Hz}$. However, when $M=12$, this rate gap increases to almost $11\,\text{bps/Hz}$. The reason is because when the number of BSs increases, we can choose the dominating paths from increased number of total paths so that we can achieve the gain obtained from path diversity.

In Fig. 13, we investigate the performance of proposed dominating PGI feedback when only one BS serves users in a cell. Although the gain obtained from the BS cooperation would not be significant in this scenario, we can still acquire accurate dominating PGI and control the inter-user interference via precoding matrix optimization in the proposed scheme. As a result, the proposed scheme achieves more than $4\,\text{dB}$ gain in the low SNR region and $3\,\text{dB}$ gain in the mid SNR region over the AoD-adaptive subspace scheme.
\\

\section{Conclusion}
In this paper, we proposed a novel feedback reduction technique for FDD-based cell-free systems. The key feature of the proposed scheme is to choose a few dominating paths among all possible propagation paths and then feed back the PGI of the chosen paths. Key observation in our work is that 1) the spatial domain channel is represented by a small number of multi-path components (AoDs and path gains) and 2) the AoDs are quite similar in the uplink and downlink channel owing to the angle reciprocity so that the BSs can acquire AoD information directly from the uplink pilot signal. Thus, by choosing a few dominating paths and only feed back the path gain of the chosen paths, we can achieve a significant reduction in the feedback overhead. We observed from the extensive simulations that the proposed scheme can achieve more than $80\%$ of feedback overhead reduction over the conventional schemes relying on the CSI feedback.

\section*{Appendix A\\ Proof of Theorem 1}
We first compute the closed-form expression of numerator of $R_{k}$ and then compute the closed-form expression of denominator of $R_{k}$. Note that the channel vector is decomposed as 
\begin{align}
\mathbf{h}_{m,k}&=\mathbf{A}_{m,k}\mathbf{g}_{m,k}=\mathbf{A}_{\Lambda_{m,k}}\mathbf{g}_{\Lambda_{m,k}}+\mathbf{A}_{\Lambda_{m,k}^{\mathsf{C}}}\mathbf{g}_{\Lambda_{m,k}^{\mathsf{C}}},\label{8.1.1.1}
\end{align}
the numerator of $R_{k}$ is given by
\begin{align}
\mathbb{E}\left[\left\lvert\sum_{m=1}^{M}\mathbf{h}_{m,k}^{\textrm{H}}\mathbf{w}_{m,k}\right\rvert^{2}\right]=&\mathbb{E}\left[\left\lvert\sum_{m=1}^{M}\mathbf{g}_{\Lambda_{m,k}}^{\textrm{H}}\mathbf{A}_{\Lambda_{m,k}}^{\textrm{H}}\mathbf{V}_{\Lambda_{m,k}}\mathbf{g}_{\Lambda_{m,k}}\right\rvert^{2}\right]+\mathbb{E}\left[\left\lvert\sum_{m=1}^{M}\mathbf{g}_{\Lambda_{m,k}^{\mathsf{C}}}^{\textrm{H}}\mathbf{A}_{\Lambda_{m,k}^{\mathsf{C}}}^{\textrm{H}}\mathbf{V}_{\Lambda_{m,k}}\mathbf{g}_{\Lambda_{m,k}}\right\rvert^{2}\right]\nonumber\\
=&\mathbb{E}\left[\left\lvert\mathbf{g}_{\Lambda_{k}}^{\textrm{H}}\mathbf{A}_{\Lambda_{k}}^{\textrm{H}}\mathbf{V}_{\Lambda_{k}}\mathbf{g}_{\Lambda_{k}}\right\rvert^{2}\right]+\mathbb{E}\left[\left\lvert\mathbf{g}_{\Lambda_{k}^{\mathsf{C}}}^{\textrm{H}}\mathbf{A}_{\Lambda_{k}^{\mathsf{C}}}^{\textrm{H}}\mathbf{V}_{\Lambda_{k}}\mathbf{g}_{\Lambda_{k}}\right\rvert^{2}\right]\nonumber\\
\stackrel{(a)}=&\mathbb{E}\left[\left\lVert\mathbf{g}_{\Lambda_{k}}\right\rVert^{4}\right]\mathbb{E}\left[\left\lvert\bar{\mathbf{g}}_{\Lambda_{k}}^{\textrm{H}}\mathbf{A}_{\Lambda_{k}}^{\textrm{H}}\mathbf{V}_{\Lambda_{k}}\bar{\mathbf{g}}_{\Lambda_{k}}\right\rvert^{2}\right]+\mathbb{E}\left[\left\lVert\mathbf{g}_{\Lambda_{k}}\right\rVert^{2}\lVert\mathbf{g}_{\Lambda_{k}^{\mathsf{C}}}\rVert^{2}\right]\mathbb{E}\left[\left\lvert\bar{\mathbf{g}}_{\Lambda_{k}^{\mathsf{C}}}^{\textrm{H}}\mathbf{A}_{\Lambda_{k}^{\mathsf{C}}}^{\textrm{H}}\mathbf{V}_{\Lambda_{k}}\bar{\mathbf{g}}_{\Lambda_{k}}\right\rvert^{2}\right]\nonumber\\
=&L(L+1)\mathbb{E}\left[\left\lvert\bar{\mathbf{g}}_{\Lambda_{k}}^{\textrm{H}}\mathbf{A}_{\Lambda_{k}}^{\textrm{H}}\mathbf{V}_{\Lambda_{k}}\bar{\mathbf{g}}_{\Lambda_{k}}\right\rvert^{2}\right]+L^{2}\mathbb{E}\left[\left\lvert\bar{\mathbf{g}}_{\Lambda_{k}^{\mathsf{C}}}^{\textrm{H}}\mathbf{A}_{\Lambda_{k}^{\mathsf{C}}}^{\textrm{H}}\mathbf{V}_{\Lambda_{k}}\bar{\mathbf{g}}_{\Lambda_{k}}\right\rvert^{2}\right],\label{8.1.2.1}
\end{align}
where $(a)$ is due to the independence of the vector norm $\left\lVert\mathbf{g}_{\Lambda_{k}}\right\rVert$ and the vector direction $\bar{\mathbf{g}}_{\Lambda_{k}}$. Since $\bar{\mathbf{g}}_{\Lambda_{k}}$ and $\bar{\mathbf{g}}_{\Lambda_{k}^{\mathsf{C}}}$ are independent, the closed-form expression of the second term in \eqref{8.1.2.1} is
\begin{align}
\mathbb{E}\left[\left\lvert\bar{\mathbf{g}}_{\Lambda_{k}^{\mathsf{C}}}^{\textrm{H}}\mathbf{A}_{\Lambda_{k}^{\mathsf{C}}}^{\textrm{H}}\mathbf{V}_{\Lambda_{k}}\bar{\mathbf{g}}_{\Lambda_{k}}\right\rvert^{2}\right]&=\mathbb{E}\left[\text{tr}\left(\bar{\mathbf{g}}_{\Lambda_{k}^{\mathsf{C}}}^{\textrm{H}}\mathbf{A}_{\Lambda_{k}^{\mathsf{C}}}^{\textrm{H}}\mathbf{V}_{\Lambda_{k}}\bar{\mathbf{g}}_{\Lambda_{k}}\bar{\mathbf{g}}_{\Lambda_{k}}^{\textrm{H}}\mathbf{V}_{\Lambda_{k}}^{\textrm{H}}\mathbf{A}_{\Lambda_{k}^{\mathsf{C}}}\bar{\mathbf{g}}_{\Lambda_{k}^{\mathsf{C}}}\right)\right]\label{8.1.3.1}\\
&=\text{tr}\left(\mathbb{E}\left[\bar{\mathbf{g}}_{\Lambda_{k}^{\mathsf{C}}}\bar{\mathbf{g}}_{\Lambda_{k}^{\mathsf{C}}}^{\textrm{H}}\right]\mathbf{A}_{\Lambda_{k}^{\mathsf{C}}}^{\textrm{H}}\mathbf{V}_{\Lambda_{k}}\mathbb{E}\left[\bar{\mathbf{g}}_{\Lambda_{k}}\bar{\mathbf{g}}_{\Lambda_{k}}^{\textrm{H}}\right]\mathbf{V}_{\Lambda_{k}}^{\textrm{H}}\mathbf{A}_{\Lambda_{k}^{\mathsf{C}}}\right)\label{8.1.3.2}\\
&=\frac{1}{L^{2}}\left\lVert\mathbf{A}_{\Lambda_{k}^{\mathsf{C}}}^{\textrm{H}}\mathbf{V}_{\Lambda_{k}}\right\rVert_{\text{F}}^{2}.\label{8.1.3.4}
\end{align} 
Whereas, the closed-form expression of the first term in \eqref{8.1.2.1} is not easy to compute. To address this issue, we use the following lemma. 
\begin{lemma}
Let $\mathbf{A}$ be a $L\times L$ matrix, $\mathbf{g}$ be a $L\times 1$ complex normal vector, and $\bar{\mathbf{g}}=\frac{\mathbf{g}}{\left\lvert\mathbf{g}\right\rvert}$. Then, 
\begin{align}\label{8.1.4}
\mathbb{E}\left[\left\lvert\bar{\mathbf{g}}^{\textup{\textrm{H}}}\mathbf{A}\bar{\mathbf{g}}\right\rvert^{2}\right]=\frac{1}{L(L+1)}\left(\lvert\textup{tr}\left(\mathbf{A}\right)\rvert^{2}+\lVert\mathbf{A}\rVert_{\textup{F}}^{2}\right).
\end{align}
\end{lemma}
\begin{proof}
Let $(i,j)$-th element of $\mathbf{A}$ be $a_{i,j}$ and $i$-th element of $\bar{\mathbf{g}}$ be $g_{i}$. Then,
\begin{align}
\mathbb{E}\left[\left\lvert\bar{\mathbf{g}}^{\textrm{H}}\mathbf{A}\bar{\mathbf{g}}\right\rvert^{2}\right]&=\mathbb{E}\Big[\Big\lvert\sum_{i,j}a_{i,j}g_{i}^{*}g_{j}\Big\rvert^{2}\Big]\label{8.1.5.1}\\
&=\mathbb{E}\Big[\Big\lvert\sum_{i}a_{i,i}\lvert g_{i}\rvert^{2}\Big\rvert^{2}\Big]+\mathbb{E}\Big[\Big\lvert\sum_{i\neq j}a_{i,j}g_{i}^{*}g_{j}\Big\rvert^{2}\Big]\label{8.1.5.3}\\
&=\sum_{i}\lvert a_{i,i}\rvert^{2}\mathbb{E}\left[\lvert g_{i}\rvert^{4}\right]+\sum_{i\neq j}a_{i,i}^{*}a_{j,j}\mathbb{E}\left[\lvert g_{i}\rvert^{2}\lvert g_{j}\rvert^{2}\right]+\sum_{i\neq j}\lvert a_{i,j}\rvert^{2}\mathbb{E}\left[\lvert g_{i}\rvert^{2}\lvert g_{j}\rvert^{2}\right]\label{8.1.5.4}\\
&\stackrel{(a)}{=}\frac{2}{L(L+1)}\sum_{i}\lvert a_{i,i}\rvert^{2}+\frac{1}{L(L+1)}\sum_{i\neq j}a_{i,i}^{*}a_{j,j}+\frac{1}{L(L+1)}\sum_{i\neq j}\lvert a_{i,j}\rvert^{2}\label{8.1.5.5}\\
&=\frac{1}{L(L+1)}\left(\Big\lvert\sum_{i}a_{i,i}\Big\rvert^{2}+\sum_{i,j}\lvert a_{i,j}\rvert^{2}\right)\label{8.1.5.6}\\
&=\frac{1}{L(L+1)}\left(\lvert\text{tr}\left(\mathbf{A}\right)\rvert^{2}+\lVert\mathbf{A}\rVert_{\text{F}}^{2}\right),\label{8.1.5.7}
\end{align}
where $(a)$ is due to the fact that $\mathbb{E}\left[\left\lvert g_{i}\right\rvert^{4}\right]=\frac{2}{L(L+1)}$ and $\mathbb{E}\left[\left\lvert g_{i}\right\rvert^{2}\right]=\mathbb{E}\left[\left\lvert g_{i}\right\rvert^{2}\left\lvert g_{j}\right\rvert^{2}\right]=\frac{1}{L(L+1)}$.
\end{proof}
\noindent By plugging the result of Lemma 3 and \eqref{8.1.3.4} into \eqref{8.1.2.1}, we get
\begin{align}
\mathbb{E}\left[\left\lvert\sum_{m=1}^{M}\mathbf{h}_{m,k}^{\textrm{H}}\mathbf{w}_{m,k}\right\rvert^{2}\right]&=\left\lvert\text{tr}\left(\mathbf{A}_{\Lambda_{k}}^{\textrm{H}}\mathbf{V}_{\Lambda_{k}}\right)\right\rvert^{2}+\left\lVert\mathbf{A}_{k}^{\textrm{H}}\mathbf{V}_{\Lambda_{k}}\right\rVert_{\text{F}}^{2}\label{8.1.6.2}\\
&=\left\lvert\sum_{m=1}^{M}\text{tr}\left(\mathbf{A}_{\Lambda_{m,k}}^{\textrm{H}}\mathbf{V}_{\Lambda_{m,k}}\right)\right\rvert^{2}+\sum_{m=1}^{M}\left\lVert\mathbf{A}_{m,k}^{\textrm{H}}\mathbf{V}_{\Lambda_{m,k}}\right\rVert_{\text{F}}^{2}.\label{8.1.6.3}
\end{align}

Next, since $\mathbf{g}_{m,k}$ and $\mathbf{g}_{\Lambda_{m,j}}$ are independent, the denominator of $R_{k}$ can be obtained similarly to \eqref{8.1.3.1}--\eqref{8.1.3.4} as
\begin{align}
\sum_{j\neq k}^{K}\mathbb{E}\left[\left\lvert\sum_{m=1}^{M}\mathbf{h}_{m,k}^{\textrm{H}}\mathbf{w}_{m,j}\right\rvert^{2}\right]&=\sum_{j\neq k}^{K}\mathbb{E}\left[\left\lvert\sum_{m=1}^{M}\mathbf{g}_{m,k}^{\textrm{H}}\mathbf{A}_{m,k}^{\textrm{H}}\mathbf{V}_{\Lambda_{m,j}}\mathbf{g}_{\Lambda_{m,j}}\right\rvert^{2}\right]\label{8.1.7.1}\\
&=\sum_{j\neq k}^{K}\sum_{m=1}^{M}\left\lVert\mathbf{A}_{m,k}^{\textrm{H}}\mathbf{V}_{\Lambda_{m,j}}\right\rVert_{\text{F}}^{2}.\label{8.1.7.2}
\end{align}
Combining \eqref{8.1.6.3} and \eqref{8.1.7.2}, we obtain the data rate expression in Theorem 1.

\section*{Appendix B\\ Proof of Proposition 1}
Let $\lbrace \mathbf{c}_{\hat{i}_{k}},\mathbf{u}_{1},\cdots,\mathbf{u}_{L-1}\rbrace$ be the orthonormal basis of $\mathbb{C}^{L}$. Also, let $\mathbf{U}=\left[\mathbf{u}_{1},\cdots,\mathbf{u}_{L-1}\right]\in\mathbb{C}^{L\times (L-1)}$. Then, the null space of $\mathbf{c}_{\hat{i}_{k}}$ can be represented as $\lbrace \mathbf{U}\boldsymbol{\alpha}\mid \left\lVert\boldsymbol{\alpha}\right\rVert=1\rbrace$ where $\boldsymbol{\alpha}$ is isotropically distributed on the $(L-1)$-dimensional unit sphere. Hence, we have
\begin{align}
\mathbb{E}\left[\mathbf{s}\mathbf{s}^{\textrm{H}}\mid\mathbf{c}_{\hat{i}_{k}}\right]&=\mathbf{U}\mathbb{E}\left[\boldsymbol{\alpha}\boldsymbol{\alpha}^{\textrm{H}}\right]\mathbf{U}^{\textrm{H}}=\frac{1}{L-1}\mathbf{U}\mathbf{U}^{\textrm{H}}\stackrel{(a)}{=}\frac{1}{L-1}\left(\mathbf{I}_{L}-\mathbf{c}_{\hat{i}_{k}}\mathbf{c}_{\hat{i}_{k}}^{\textrm{H}}\right),
\end{align}
where $(a)$ is due to the fact that $\mathbf{I}_{L}=\left[\mathbf{c}_{\hat{i}_{k}}\,\mathbf{U}\right]\left[\mathbf{c}_{\hat{i}_{k}}\,\mathbf{U}\right]^{\textrm{H}}=\mathbf{c}_{\hat{i}_{k}}\mathbf{c}_{\hat{i}_{k}}^{\textrm{H}}+\mathbf{U}\mathbf{U}^{\textrm{H}}$.

\section*{Appendix C\\ Proof of Theorem 3}
Recall that the precoding matrix $\mathbf{V}_{\Lambda_{k}}$ is obtained from the de-vectorization of $\mathbf{x}_{\Lambda_{k}}=\frac{\mathbf{u}_{k,\text{max}}}{\left\lVert\mathbf{u}_{k,\text{max}}\right\rVert}$. Here, $\mathbf{u}_{k,\text{max}}$ is the eigenvector corresponding to the largest eigenvalue of $\mathbf{W}_{k}^{-1}\mathbf{U}_{k}$ where
\begin{align}\label{R.2.3}
\mathbf{U}_{k}&=\boldsymbol{\mu}_{\Lambda_{k}}\boldsymbol{\mu}_{\Lambda_{k}}^{\textrm{H}}+\mathbf{I}_{L}\otimes\mathbf{A}_{k}\mathbf{A}_{k}^{\textrm{H}}\\
\mathbf{W}_{k}&=\sum_{j\neq k}^{K}\mathbf{I}_{L}\otimes\mathbf{A}_{j}\mathbf{A}_{j}^{\textrm{H}}+\sigma_{n}^{2}\mathbf{I}_{NL}=\mathbf{I}_{L}\otimes \left(\boldsymbol{\Phi}_{k}\boldsymbol{\Phi}_{k}^{\textrm{H}}+\sigma_{n}^{2}\mathbf{I}_{N}\right),
\end{align}
where $\boldsymbol{\Phi}_{k}=\left[\mathbf{A}_{j},\, j\neq k\right]\in\mathbb{C}^{N\times (K-1)P}$. By using the Woodbury matrix identity, we obtain
\begin{align}\label{R.2.4}
\mathbf{W}_{k}^{-1}&=\mathbf{I}_{L}\otimes\left(\frac{1}{\sigma_{n}^{2}}\mathbf{I}_{N}-\frac{1}{\sigma_{n}^{2}}\boldsymbol{\Phi}_{k}\left(\boldsymbol{\Phi}_{k}^{\textrm{H}}\boldsymbol{\Phi}_{k}+\sigma_{n}^{2}\mathbf{I}_{(K-1)P}\right)^{-1}\boldsymbol{\Phi}_{k}^{\textrm{H}}\right)\\
&\stackrel{(a)}{=}\frac{1}{\sigma_{n}^{2}}\mathbf{I}_{L}\otimes\left(\mathbf{I}_{N}-\frac{1}{\left(1+\sigma_{n}^{2}\right)}\boldsymbol{\Phi}_{k}\boldsymbol{\Phi}_{k}^{\textrm{H}}\right)\label{R.2.4.3},
\end{align}
where $(a)$ is due to the fact that $\boldsymbol{\Phi}_{k}^{\textrm{H}}\boldsymbol{\Phi}_{k}=\mathbf{I}_{(K-1)P}$. Thus, we get 
\begin{align}\label{R.2.5}
\mathbf{W}_{k}^{-1}\mathbf{U}_{k}&=\frac{1}{\sigma_{n}^{2}}\left(\mathbf{I}_{L}\otimes \left(\mathbf{I}_{N}-\frac{1}{1+\sigma_{n}^{2}}\boldsymbol{\Phi}_{k}\boldsymbol{\Phi}_{k}^{\textrm{H}}\right)\right)\left(\boldsymbol{\mu}_{\Lambda_{k}}\boldsymbol{\mu}_{\Lambda_{k}}^{\textrm{H}}+\mathbf{I}_{L}\otimes\mathbf{A}_{k}\mathbf{A}_{k}^{\textrm{H}}\right)\\
&\stackrel{(a)}{=}\frac{1}{\sigma_{n}^{2}}\left(\boldsymbol{\mu}_{\Lambda_{k}}\boldsymbol{\mu}_{\Lambda_{k}}^{\textrm{H}}+\mathbf{I}_{L}\otimes\mathbf{A}_{k}\mathbf{A}_{k}^{\textrm{H}}\right)\label{R.2.5.3}
\end{align} 
where $(a)$ is due to the fact that $\boldsymbol{\Phi}_{k}$ is orthogonal to $\boldsymbol{\mu}_{\Lambda_{k}}$ and $\mathbf{A}_{k}$. From \eqref{R.2.5.3}, we observe that $\mathbf{x}_{\Lambda_{k}}$ is the eigenvector of $\mathbf{U}_{k}$. Consequently, $\mathbf{x}_{\Lambda_{k}}$ is in the column space of $\mathbf{U}_{k}$ which is orthogonal to the column space of $\mathbf{I}_{L}\otimes \mathbf{A}_{j}$ for every $j\neq k$. Thus, the rate in \eqref{3.3.2.2} can be re-expressed as
\begin{align}\label{R.2.6}
R_{k}^{(\text{ideal})}&=\log_{2}\left(1+\frac{\left\lvert\boldsymbol{\mu}_{\Lambda_{k}}^{\textrm{H}}\mathbf{x}_{\Lambda_{k}}\right\rvert^{2}+\mathbf{x}_{\Lambda_{k}}^{\textrm{H}}\left(\mathbf{I}_{L}\otimes \mathbf{A}_{k}\mathbf{A}_{k}^{\textrm{H}}\right)\mathbf{x}_{\Lambda_{k}}}{\sum_{j\neq k}^{K}\mathbf{x}_{\Lambda_{j}}^{\textrm{H}}\left(\mathbf{I}_{L}\otimes \mathbf{A}_{k}\mathbf{A}_{k}^{\textrm{H}}\right)\mathbf{x}_{\Lambda_{j}}+\sigma_{n}^{2}}\right)\\
&=\log_{2}\left(1+\frac{1}{\sigma_{n}^{2}}\mathbf{x}_{\Lambda_{k}}^{\textrm{H}}\left(\boldsymbol{\mu}_{\Lambda_{k}}\boldsymbol{\mu}_{\Lambda_{k}}^{\textrm{H}}+\mathbf{I}_{L}\otimes \mathbf{A}_{k}\mathbf{A}_{k}^{\textrm{H}}\right)\mathbf{x}_{\Lambda_{k}}\right)\\
&=\log_{2}\left(1+\frac{1}{\sigma_{n}^{2}}\lambda_{k,\text{max}}\right),\label{R.2.6.3}
\end{align}
where $\lambda_{k,\text{max}}$ is the largest eigenvalue of $\mathbf{U}_{k}$. In the following lemma, we provide $\lambda_{k,\text{max}}$ as a function of $L$.
\begin{lemma}
The largest eigenvalue $\lambda_{k,\text{max}}$ of $\mathbf{U}_{k}$ is $L+1$
\end{lemma}
\begin{proof}
We show that $\boldsymbol{\mu}_{\Lambda_{k}}$ is an eigenvector of $\mathbf{U}_{k}$ corresponds to $\lambda_{k,\text{max}}=L+1$. Note that
\begin{align}\label{R.2.7}
\mathbf{U}_{k}&=\boldsymbol{\mu}_{\Lambda_{k}}\boldsymbol{\mu}_{\Lambda_{k}}^{\textrm{H}}+\mathbf{I}_{L}\otimes\mathbf{A}_{k}\mathbf{A}_{k}^{\textrm{H}}\\
&=\left[\boldsymbol{\mu}_{\Lambda_{k}},\,\mathbf{I}_{L}\otimes\mathbf{A}_{k}\right]\left[\boldsymbol{\mu}_{\Lambda_{k}},\,\mathbf{I}_{L}\otimes\mathbf{A}_{k}\right]^{\textrm{H}}
\end{align}
It is worth mentioning that the columns of $\mathbf{I}_{L}\otimes\mathbf{A}_{k}$ are mutually orthonormal. Also, since $\boldsymbol{\mu}_{\Lambda_{k}}=\text{vec}\left(\mathbf{A}_{\Lambda_{k}}\right)$, $\boldsymbol{\mu}_{\Lambda_{k}}$ can be expressed as a linear combination of the columns of $\mathbf{I}_{L}\otimes \mathbf{A}_{k}$ (i.e., $\boldsymbol{\mu}_{\Lambda_{k}}=\sqrt{L}\left(\mathbf{I}_{L}\otimes \mathbf{A}_{k}\right)\boldsymbol{\alpha}$ for some $\boldsymbol{\alpha}\in\mathbb{C}^{LP}$, $\left\lVert\boldsymbol{\alpha}\right\rVert=1$). Hence, the columns of $\mathbf{I}_{L}\otimes \mathbf{A}_{k}$ form an orthonormal basis of the column space of $\mathbf{U}_{k}$.  Note that
\begin{align}\label{R.2.8}
\mathbf{U}_{k}\boldsymbol{\mu}_{\Lambda_{k}}&=\left(\boldsymbol{\mu}_{\Lambda_{k}}\boldsymbol{\mu}_{\Lambda_{k}}^{\textrm{H}}+\mathbf{I}_{L}\otimes\mathbf{A}_{k}\mathbf{A}_{k}^{\textrm{H}}\right)\sqrt{L}\left(\mathbf{I}_{L}\otimes\mathbf{A}_{k}\right)\boldsymbol{\alpha}\\
&=L\boldsymbol{\mu}_{\Lambda_{k}}+\sqrt{L}\left(\mathbf{I}_{L}\otimes\mathbf{A}_{k}\right)\boldsymbol{\alpha}\\
&=\left(L+1\right)\boldsymbol{\mu}_{\Lambda_{k}}
\end{align}
Thus, $\boldsymbol{\mu}_{\Lambda_{k}}$ is the eigenvector corresponding to the eigenvalue $L+1$ of $\mathbf{U}_{k}$. Now, let $\mathbf{v}$ be an eigenvector corresponding to the eigenvalue $\lambda$ of $\mathbf{U}_{k}$. Since $\mathbf{v}$ is in the column space of $\mathbf{U}_{k}$, it can be expressed as $\mathbf{v}=\left(\mathbf{I}_{L}\otimes \mathbf{A}_{k}\right)\boldsymbol{\beta}$ for some $\boldsymbol{\beta}\in\mathbb{C}^{LP}$, $\left\lVert\boldsymbol{\beta}\right\rVert=1$. Then we get
\begin{align}\label{R.2.9}
\lambda&=\mathbf{v}^{\textrm{H}}\mathbf{U}_{k}\mathbf{v}\\
&=\boldsymbol{\beta}^{\textrm{H}}\left(\mathbf{I}_{L}\otimes \mathbf{A}_{k}^{\textrm{H}}\right)\left(\boldsymbol{\mu}_{\Lambda_{k}}\boldsymbol{\mu}_{\Lambda_{k}}^{\textrm{H}}+\mathbf{I}_{L}\otimes\mathbf{A}_{k}\mathbf{A}_{k}^{\textrm{H}}\right)\left(\mathbf{I}_{L}\otimes \mathbf{A}_{k}\right)\boldsymbol{\beta}\\
&=\boldsymbol{\beta}^{\textrm{H}}\left(\mathbf{I}_{L}\otimes \mathbf{A}_{k}^{\textrm{H}}\right)\left(L\left(\mathbf{I}_{L}\otimes \mathbf{A}_{k}\right)\boldsymbol{\alpha}\boldsymbol{\alpha}^{\textrm{H}}\left(\mathbf{I}_{L}\otimes\mathbf{A}_{k}^{\textrm{H}}\right)+\mathbf{I}_{L}\otimes\mathbf{A}_{k}\mathbf{A}_{k}^{\textrm{H}}\right)\left(\mathbf{I}_{L}\otimes \mathbf{A}_{k}\right)\boldsymbol{\beta}\\
&=\boldsymbol{\beta}^{\textrm{H}}\left(L\boldsymbol{\alpha}\boldsymbol{\alpha}^{\textrm{H}}+\mathbf{I}_{NL}\right)\boldsymbol{\beta}\\
&=L\left\lvert\boldsymbol{\alpha}^{\textrm{H}}\boldsymbol{\beta}\right\rvert^{2}+1
\end{align}
Hence, $\lambda_{k,\text{max}}$ is obtained when $\boldsymbol{\beta}=\boldsymbol{\alpha}$ and thus, we get $\lambda_{k,\text{max}}=L+1$.
\end{proof}
\noindent By using Lemma 4, $R_{k}^{(\text{ideal})}$ in \eqref{R.2.6.3} is re-expressed as
\begin{align}\label{R.2.10}
R_{k}^{(\text{ideal})}=\log_{2}\left(1+\frac{L+1}{\sigma_{n}^{2}}\right)
\end{align} 
Finally, combining \eqref{R.2.10} and the result of Theorem 2, we obtain the desired result.

\bibliography{refs}
\bibliographystyle{IEEEtran}
\clearpage 

\end{document}